\documentclass[18pt,english]{article}
\usepackage[utf8]{inputenc}
\usepackage[a4paper, total={6in, 9in}]{geometry}
\usepackage[pdftex]{graphicx}
\usepackage[table]{xcolor}
\usepackage{float}
\usepackage{hhline}
\usepackage{amssymb}
\usepackage{biblatex}
\usepackage{authblk}
\usepackage{amsmath}
\usepackage{amsthm}
\usepackage{eurosym}
\usepackage{hyperref}
\usepackage{hhline}
\newtheorem{definition}{Definition}
\newtheorem{proposition}{Proposition}
\newtheorem{theorem}{Theorem}
\newtheorem{example}{Example}
\newtheorem{lemma}{Lemma}
\newtheorem{corollary}{Corollary}
\newcommand{\commentout}[1]{}

\bibliography{bibl}

\title{
Maximizing Social Welfare with Side Payments
}

\author[1]{Ivan Geffner\thanks{Work done while the author was at FOCAL (CMU).}}
\author[2,3]{Caspar Oesterheld}
\author[2,3]{Vincent Conitzer}

\affil[1]{Utrecht University}
\affil[2]{Carnegie Mellon University}
\affil[3]{FOCAL - Foundations of Cooperative AI Lab}

\date{}

\begin{document}

\maketitle

\begin{abstract}
\commentout{
We study a setting where players can simultaneously commit to perform outcome-dependent side payments before playing an action in a normal-form game. 
Jackson and Wilkie~\cite{jackson2005endogenous} 
showed that, in the one-shot model, the ability of side contracting may not always be beneficial to the players. In fact, they give an example of a game with a Pareto optimal Nash equilibrium in which all the resulting equilibria that arise from the inclusion of side payments are inefficient. The intuition for this phenomenon is that players can end up in a prisoner's dilemma scenario, where each player can selfishly use her commitment power to change the equilibrium structure of the game for her benefit at the expense of the other players' welfare.

In order to avoid such outcomes, we consider a variant of Jackson and Wilkie's model where players can commit to small
side payments during successive rounds. We show that, under this model, if players start with a normal-form game $\Gamma$ and a non-degenerate Nash equilibrium $\vec{\sigma}$, they can implement all welfare-maximizing utility profiles that strictly Pareto improve $\vec{\sigma}$. 
}
We examine normal-form games in which players may \emph{pre-commit} to outcome-contingent transfers before choosing their actions.  In the one-shot version of this model, Jackson and Wilkie~\cite{jackson2005endogenous} showed that side contracting can backfire: even a game with a Pareto-optimal Nash equilibrium can devolve into inefficient equilibria once unbounded, simultaneous commitments are allowed.  The root cause is a prisoner’s dilemma effect, where each player can exploit her commitment power to reshape the equilibrium in her favour, harming overall welfare.

To circumvent this problem we introduce a \emph{staged-commitment} protocol.  Players may pledge transfers only in small, capped increments over multiple rounds, and the phase continues only with unanimous consent.  We prove that, starting from any finite game~$\Gamma$ with a non-degenerate Nash equilibrium~$\vec{\sigma}$, this protocol implements every welfare-maximising payoff profile that \emph{strictly} Pareto-improves $\vec{\sigma}$.  Thus, gradual and bounded commitments restore the full efficiency potential of side payments while avoiding the inefficiencies identified by Jackson and Wilkie.
\end{abstract}

\section{Introduction}\label{sec:intro}

Nash equilibrium is a central tool in game theory, yet it abstracts away an important feature of many real-world settings: players often negotiate before they move. Outside of the “vacuum” assumed in normal-form games, agents can coordinate and steer away from inefficient outcomes, frequently by promising side payments that realign incentives and support mutually preferred strategies.
Governments routinely illustrate this logic: they subsidize firms that curb greenhouse-gas emissions, compensating them for the costs of cooperation. Similarly, criminal defendants often secure plea bargains that exchange monetary restitution for lighter sentences. For a more concrete example, consider the two-player Prisoner’s Dilemma whose payoffs appear in the left panel of Table 1.

\begin{example}\label{ex-prisoner}

\[  
\begin{array}{|c||c|c|}
\hline
 & C & D \\
 \hhline{|=#=|=|}
 C & (0,0) & (-2, 1) \\
 \hline
 D & (1, -2) & (-1,-1).\\
 \hline
\end{array}
\qquad \qquad
\begin{array}{|c||c|c|}
\hline
 & C & D \\
 \hhline{|=#=|=|}
 C & (0,0) & (-1, 0) \\
 \hline
 D & (0, -1) & (-1,-1)\\
 \hline
\end{array}
\]
\end{example}

As usual, the only Nash equilibrium is for both players to defect (D), but both players defecting is Pareto dominated by the outcome in which everyone cooperates (C). Now suppose that there is a pre-play stage in which both players can submit binding side payment contracts before playing an action in the example above. Then, each player can commit to pay $1$ utility to the other player if the other player cooperates. By doing so, they effectively transform the original payoff matrix into the one on the right,
making $(C,C)$ a Nash equilibrium of the resulting game. One may wonder if this example can be generalized to arbitrary normal-form games, and if the inclusion of side payments always results in an increase of social welfare. The intuition is that if player 
$i$’s gain from some action 
$a$ is smaller than the loss it imposes on others, those others can offer compensation below their loss yet sufficient to induce 
$i$ to switch. However, Jackson and Wilkie~\cite{jackson2005endogenous} show that such contracts are not necessarily best responses when all players can write them; indeed, transfers can be used strategically to condition opponents’ behavior. The point is illustrated by an asymmetric variant of Chicken (left panel below). 

\begin{example}\label{ex:chicken}

\[\begin{array}{|c||c|c|}
\hline
 & \text{Swerve} & \text{Straight} \\
 \hhline{|=#=|=|}
 \text{Swerve} & (0,0) & (1,2) \\
 \hline
 \text{Straight} & (2, 0) & (-10, -10)\\
 \hline
\end{array}
\qquad \qquad
\begin{array}{|c||c|c|}
\hline
 & \text{Swerve} & \text{Straight} \\
 \hhline{|=#=|=|}
 \text{Swerve} & (0,0) & (-19,22) \\
 \hline
 \text{Straight} & (2, 0) & (-10, -10)\\
 \hline
\end{array}
\]
\end{example}

In this game, (Swerve, Straight) is the outcome that maximizes social welfare. However, the first player (i.e., the row player) can commit to pay 20 utility on (Swerve, Straight), which transforms the payoff matrix into the one on the right.
Now, going straight is a dominant strategy for player 1, and therefore (Straight, Swerve) becomes the only Nash equilibrium. Thus, committing this way guarantees that player 1 gets 2 utility, as opposed to 1 utility in the (Swerve, Straight) outcome of the original game.
Another problem that arises
with the inclusion of side payments is the issue of outcome fairness.
For instance, consider a normal-form game for two players in which each one of them can play either $A$ or $B$ and the payoff matrix is as follows.

\begin{example}\label{ex:unfair}
$$\begin{array}{|c||c|c|}
\hline
 & A & B \\
 \hhline{|=#=|=|}
 A & (0,0) & (-2, -2) \\
 \hline
 B & (-2,-2) & (10, -3)\\
 \hline
\end{array}$$
\end{example}

In this example, $(A,A)$ is the only Nash equilibrium by iterated strict dominance. 
With the inclusion of side payments, we might hope that they can reach 
an
agreement on $(B,B)$, with a transfer from player 1 to player 2, so that both players get strictly positive utility. However, if $(B,B)$ is a Nash equilibrium after committing, it would also be so if player $2$'s utility in $(B,B)$ was $-2$. Since it is not beneficial for player $1$ to offer more than what is strictly necessary, this would imply that player $2$ gets less utility than in the original equilibrium.

The underlying problem is that pre-play side payments let players drastically reshape the game’s equilibrium set without any safeguard against defection during the commitment stage. We therefore consider payment protocols that build in such safeguards. One natural device is a rollback option: after the commitment phase, a second pre-play stage is held in which each player votes either to accept the proposed side payments or to void them all. If at least one player votes to roll back, all contracts are annulled and the original game is played; otherwise, the transfers remain in force. Under this rollback model, one can implement every welfare-maximizing utility profile that Pareto-improves on at least one equilibrium of the base game. For example, in Example~\ref{ex:unfair}, suppose we wish to construct a subgame perfect equilibrium in which player 1 obtains 
$4$ utility and player $2$ obtains $3$ (the maximum social welfare there is $7$). We can proceed as follows.
\begin{enumerate}
    \item First, player $1$ commits to pay $6$ utility to player $2$ if the outcome is $(B,B)$.
    \item If player $1$ made the commitment above, both players agree on the commitments and play $(B,B)$. Otherwise, they both vote for a rollback and play $(A,A)$.
\end{enumerate}

One can verify that this strategy is a subgame-perfect equilibrium, and that the same construction extends to finite games with any number of players. However, rollback power is typically unrealistic, as it presumes that players can nullify others’ binding contracts.

We propose a staged-commitment protocol that achieves nearly the same outcomes as the rollback benchmark. Instead of a single pre-play stage, players may make publicly observed, binding transfers in small (capped) increments across multiple stages; transfers can be directed to other players or to an uninvolved sink (i.e., \emph{burning money}). Progress to the next stage requires unanimous consent; if any player refuses, the commitment phase ends and play proceeds in the game augmented by all transfers made so far. The intended interpretation is a gradual bargaining process in which players credibly signal goodwill via incremental commitments. This is reminiscent of disarmament negotiations, where arsenals are reduced step by step to avoid opportunistic attacks (see \cite{Deng17:Disarmament,Deng18:Disarmament}). Our main result shows that, under this protocol, the implementable utility profiles coincide with those under rollback, except for a few degenerate cases. For example, in Example~\ref{ex:unfair}, with a per-round cap of $1$, there exists a subgame-perfect equilibrium yielding $1$ utility to player $1$ and $3$ utility to player $2$: player $1$ commits $1$ utility to player~2 on $(B,B)$ in each of six consecutive stages, terminates the commitment phase, and both then play $(B,B)$; any deviation triggers immediate termination and reversion to $(A,A)$.

\begin{center}
    \includegraphics[scale=0.7]{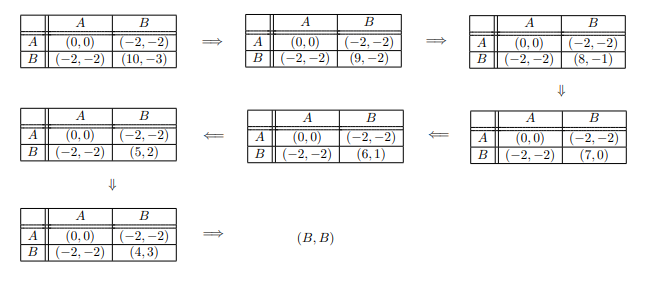}
\end{center}

One can verify that this strategy is a subgame-perfect equilibrium. After any deviation in the commitment phase, the per-round cap of $1$ ensures that $(A,A)$ remains a Nash equilibrium of the resulting game. In the continuation, \(A\) is therefore a best response (when the opponent plays $A$), and the path reverts to $(A,A)$. Any deviation thus yields payoffs $(0,0)$, strictly worse than the target $(4,3)$.

The approach for arbitrary normal-form games follows the same lines as the example above. Players commit to side payments in a way such that, after each commitment round, the induced game admits a \emph{punishment equilibrium}: a Nash equilibrium that gives each player a lower payoff than the targeted full-game path (the commitment history plus the continuation play). The main idea is that, if a player ever defects from the full-game strategy, the remaining players can hurt them by ending the commitment phase and playing this punishment equilibrium. Even though the defecting player can theoretically break the punishment equilibrium by placing additional side payments, the upper bound on the payments that can be made in a single round guarantees that the punishment equilibrium persists regardless of how the players commit, and therefore that players have no incentives to defect. The general proof proceeds by cases, the most delicate being when the social-welfare-maximizing outcome already occurs with positive probability in the original equilibrium (and hence in the punishment equilibrium). Our main technical contribution is a protocol that sequences commitments so that the desired outcome eventually becomes a Nash equilibrium while preserving the punishment structure at every history.

Our construction for arbitrary finite normal-form games mirrors the example above. We choose transfers so that, after every commitment round, the induced game admits a \emph{punishment equilibrium}: a Nash equilibrium that gives each player a lower payoff than the targeted full-game path (the commitment history plus the continuation play). If any player deviates, all players terminate the commitment phase and play this punishment equilibrium. Although a deviator could, in principle, try to overturn the punishment equilibrium by adding further transfers, the per-round cap ensures that the punishment equilibrium survives after any single round, so no profitable deviation exists. 

\subsection{Implementability vs. Convergence}

It is important to note that implementability alone does not guarantee that players will converge to a welfare-maximizing profile that Pareto-improves on the original equilibrium. One might worry that (i) the constructions are too intricate for agents to discover independently; (ii) selecting among Pareto-improving profiles creates a coordination problem reminiscent of \emph{Battle of the Sexes}, where all players prefer to agree but each favors receiving a larger share of the surplus; and (iii) introducing a commitment phase could generate new inefficient equilibria. In practice, these concerns are limited. First, a central aim of this paper is to provide explicit, usable implementations for each target equilibrium. Second, coordination can be handled by fixing a \emph{default} fallback payoff profile to be used when there is no explicit agreement (e.g., equal division of the surplus, or allocation proportional to players’ baseline utilities). Third, any inefficient equilibria arising from the pre-play phase cannot stray far from the original ones, since any player can unilaterally terminate immediately after the first (capped) commitments; in particular, when the base game has a unique equilibrium, all such inefficient equilibria are Pareto-dominated by the welfare-maximizing alternatives. While this does not theoretically guarantee convergence, there is ample evidence that, when incentives are aligned, rational players tend to avoid inefficient outcomes~\cite{gold2020team, bardsley2010explaining}. Proposed explanations include \emph{team reasoning}\cite{sugden1993thinking, bacharach1999interactive, bacharach2006beyond}, where players select the option best for the group, and evidential decision theory\cite{gibbard1978counterfactuals}, under which a player expecting others to reason similarly optimally chooses the Pareto-efficient profile.

\subsection{Related Work}

Our work is closest to Jackson and Wilkie~\cite{jackson2005endogenous}. They study a single round of simultaneous, unbounded commitments before actions are chosen in the base game and do not allow money burning. Within this model, they characterize the set of supportable equilibrium payoffs and show that side payments can reduce social welfare: there are games with Pareto-efficient outcomes in which every equilibrium induced by side payments is inefficient. Sauerberg and Oesterheld~\cite{Sauerberg24:Computing} analyze a closely related setting in which players commit to transfers while simultaneously choosing their actions (rather than acting after observing others’ commitments), and Gupta and Schewe consider a similar two-player model where the first mover commits both to an action and to a \emph{bribe} vector before the second moves. Although these modeling choices may appear minor, they lead to substantially different equilibrium sets from ours.

A separate line of work, beginning with Guttman~\cite{guttman1978understanding, guttman1987non}, Danziger and Schnytzer~\cite{danziger1991implementing}, Guttman and Schnytzer~\cite{guttman1992solution}, and Varian~\cite{varian1994sequential, varian1994solution}, achieves efficiency in public-good environments via matching schemes. The available contracts there are much more restrictive (for example, payments proportional to particular types of actions), and thus differ substantially from the side-payment flexibility we study.

Another closely related literature examines {\em disarmament games}. These also involve iterative commitments, but players directly shrink their own strategy sets rather than transferring utility. As in our setting, participants cannot reduce ``too much'' at once without becoming vulnerable, yet Deng and Conitzer~\cite{Deng17:Disarmament,Deng18:Disarmament} show that gradual, small-step commitments can achieve the desired targets. This aligns with the results presented in this paper.

Payments used to shape outcomes also arise in several other directions. In mechanism design, payments are chosen to induce truthful reports, but the underlying game is designer-specific rather than arbitrary. For instance, in Bayesian allocation problems that maximize reported total value, the VCG mechanism~\cite{vickrey1961counterspeculation, clarke1971multipart, groves1973incentives} makes truthful reporting optimal. The principal-agent literature \cite{Lambert1986, Demski1987, Stoughton1993, Barron2003, Feess2004, Gromb2007, Malcomson2009} studies settings in which one player (the principal) can commit to pay another player (the agent) in a way that depends on what action the agent takes. However, payments typically go only one way and this literature focuses on issues absent from our setting, such as imperfect observability of the agent's action or differing risk preferences between agent and principal.
Meanwhile, other results involve commitments to outcome-dependent payments for arbitrary games.  In Monderer and Tennenholtz's $k$-implementation paper~\cite{Monderer04:K}, an outside party commits to outcome-dependent payments to the players in the game in order to transform a Nash equilibrium into a dominant strategy equilibrium. Their approach is similar to ours in the sense that these payments are always promised on outcomes outside of the resulting equilibrium, which means that the payer ultimately does not have to put her own resources on the line. In later work on {\em internal implementation}~\cite{Anderson10:Internal}, the party with the ability to commit to such payments is modeled as one of the players in the game.

Many relaxations of Nash equilibrium introduce a {\em mediator}. The standard version communicates only via {\em signals}: correlated equilibrium~\cite{aumann1974subjectivity} can be viewed as a mediator sending (possibly correlated) recommendations~\cite{lehrer2010signaling, arieli2023mediated, arieli2022bayesian}, and this perspective also underlies {\em Bayesian persuasion}~\cite{kamenica2011bayesian, arieli2022bayesian}. More general models allow multi-round communication with a mediator~\cite{forges1986approach,ben2003cheap,abraham2006distributed,abraham2019implementing,abraham2008lower,geffner2023lower,gerardi2004unmediated, geffner2024communication}. Stronger mediators can do more than signal: empowering the mediator to act on players’ behalf~\cite{Rozenfeld07:Routing, Monderer09:Strong} or to expand action sets~\cite{geffner2024making} yields equilibria resilient to coalitional deviations. \emph{Program equilibrium}~\cite{McAfee1984, Howard1988, Rubinstein1998, Tennenholtz04:Program, Kalai07:Commitment, Oesterheld18:Robust, Critch2019, cooper2024characterising} can likewise be seen as mediated interaction enabling players to commit to and share strategies. In a way, this is reminiscent of our setting since it requires players to commit (in this case, to a strategy instead of side payments) and allows players to converge on more efficient equilibria that would otherwise be impossible. In particular, it can be shown that, under this model, players can reach cooperation on the prisoner's dilemma. Relatedly, mediators that permit simulating opponents can produce similar effects~\cite{kovarik2023game}. In this light, our model can be interpreted as a mediated protocol: players submit side payments to a mediator, who can terminate the commitment phase (triggering play in the base game) and enforce transfers at the realized outcome.

The remainder of the paper is organized as follows. Section~\ref{sec:defs} formally introduces the model, notation, and basic definitions. Section~\ref{sec:main-results} states our main results, and Section~\ref{sec:thm-bot} provides proofs for the simpler cases. We conclude in Section~\ref{sec:conclusion}. Additional technical material appears in Appendices~\ref{sec:arb-games}, \ref{sec:2pl-2act}, \ref{sec:thm-main}, and \ref{sec:thm-strong}.

\section{Basic Definitions and Notation}\label{sec:defs}

\subsection{Normal-form games}

Given a normal-form game $\Gamma = (P, A, U)$, we denote the set of players by $P := \{1,2,\ldots, n\}$, the set of action profiles by $A := A_1 \times \ldots \times A_n$, and the utility function by $U := (u_1, \ldots, u_n)$, where $u_i : A \to \mathbb{R}$ denotes the utility function of player $i$. Unless stated otherwise, we will always label the player $i$'s actions as $A_i := 
\{a_i^1, a_i^2, \ldots, a_i^{N_i}\}$ for some $N_i \in \mathbb{N}$.
Given a strategy profile $\vec{\sigma}$, we denote by $\sigma_i^j$ the probability that action $a_i^j$ is played in $\sigma_i$, by $\sigma_i(A_i)$ the set of actions of player $i$ that are played with non-zero probability in $\sigma_i$, by $u_i(\vec{\sigma})$ the expected utility that player $i$ gets when $\vec{\sigma}$ is played, and by $w(\vec{\sigma})$ the expected social welfare of the players, which is defined as $w(\vec{\sigma}) := \sum_{i = 1}^n u_i(\vec{\sigma})$. We also denote by $w_{\max}(\Gamma)$ the maximum social welfare of the game, which is defined as $w_{\max}(\Gamma) := \sup_{\vec{\sigma} \in \Delta A} \{w(\vec{\sigma})\}.$ Note that, since $w(\vec{\sigma})$ is a weighted average of the social welfare of pure strategy profiles, the maximum is always achieved by a pure strategy profile. Thus, we have that $w_{\max}(\Gamma) := \max_{\vec{a} \in A} \{w(\vec{a})\}$.

Given two normal form games $\Gamma = (P,A,U)$ and $\Gamma' = (P', A', U')$, we define their distance $d(\Gamma, \Gamma')$ as follows:

$$d(\Gamma, \Gamma') := 
\left\{
\begin{array}{ll}
\ \ +\infty & \mbox{if } P \not= P' \mbox{ or } A \not = A'\\
\displaystyle \max_{i \in P, \ \vec{a} \in A} \{|u_i(\vec{a}) - u'_i(\vec{a})|\} & \mbox{otherwise.}
\end{array}
\right.$$

Intuitively, the distance between normal-form games is infinite if they have different structures (i.e., different number of players or different action sets). Otherwise, it is the maximum difference of utilities that a player can receive on a single outcome. Note that this means that we can view the set of normal-form games as a metric space, and therefore as a topological space too. For the rest of the paper, we will be implicitly assuming that the metric and the topology of the space of normal-form games that we consider are the ones induced by the distance $d$ above.

\subsection{$\delta$-commitment games}

Given a normal-form game $\Gamma$ and a strictly positive real number $\delta$, we denote by $\Gamma^\delta$ the \emph{$\delta$-commitment extension} of $\Gamma$. An instance of $\Gamma^\delta$ runs as follows.
\begin{enumerate}
    \item \textbf{Step 1:} Each player decides how much utility it commits to pay for each possible outcome. The recipients of the paid utility can be either other players or $\bot$ (the utility in this case is burned and is not given to other players). The total utility paid in each outcome cannot exceed $\delta$.
    \item \textbf{Step 2:} The payments are made public and the utility functions $u_1, \ldots, u_n$ are updated according to the promised payments.
    \item \textbf{Step 1:} Each player votes to end the commitment phase or not. If all players agree on continuing with the commitment phase, they go back to Step 1. Otherwise, they proceed to Step 4.
    \item \textbf{Step 4:} Each player $i$ plays an action $a_i$ in the resulting game and gets $u_i(a_1, \ldots, a_n)$ utility (note that these utilities are updated every time that Step 1 is performed). The game ends.
\end{enumerate}

In $\Gamma^\delta$, a strategy $\sigma_i$ for player $i$ indicates, given the current state of the game, which payments $i$ should make in each round of the commitment phase, if it should vote to end the commitment phase or not, and what action to play in Step 4. We say that a strategy profile $\vec{\sigma}$ is a \emph{$\delta$-commitment equilibrium} of $\Gamma$ if $\vec{\sigma}$ is a subgame perfect equilibrium of $\Gamma^\delta$. For future reference, we consider also an extension $\Gamma^\delta_\bot$ of $\Gamma$ which is identical to $\Gamma^\delta$ except that players can only burn utility in Step 2 (they cannot transfer it to other players). We say that a strategy profile $\vec{\sigma}$ is a \emph{$(\delta, \bot)$-commitment equilibrium} of $\Gamma$ if $\vec{\sigma}$ is a subgame perfect equilibrium of $\Gamma^\delta_\bot$.

\subsection{Punishable and Non-Degenerate Equilibria}\label{sec:degenerate}

At a high level, our aim in this paper is proving that, given a normal-form game $\Gamma$, we can improve any Nash equilibrium that does not attain maximum social welfare by extending it to a $\delta$-commitment game, for some value $\delta \in \mathbb{R}$. Unfortunately, this result does not quite hold for all games. For instance, consider a normal-form game for two players with the following payoff matrix.

$$\begin{array}{|c||c|c|c|}
\hline
 & A & B & C \\
 \hhline{|=#=|=|=|}
 A & (5,5) & (0, 5) & (0,0) \\
 \hline
 B & (5,0) & (9, 2) & (7,1)\\
 \hline
 C & (0,0) & (1,7) & (6,6) \\
 \hline
\end{array}$$

In this example, $(A,A)$ is a Nash equilibrium that gives $5$ utility to both players and $(C,C)$ is a welfare-maximizing outcome that gives $6$ utility instead, but it is not a Nash equilibrium. Ideally, we would like players to converge on $(C,C)$. However, consider a $\delta$-commitment extension of $\Gamma$ and suppose that player $1$ commits to pay $\delta$ utility to player $2$ in $(A,B)$ and $\delta$ utility to $\bot$ in $(A,A)$. Then, regardless of what player $2$ does, action $B$ will be a dominant action for at least one of the players (it will be for player $1$ unless player $2$ pays $\delta$ utility to player $1$ in $(A,A)$, but in this case it will be for player $2$). This means that, if player $1$ performs this move and then stops the commitment phase right after, she can guarantee that the only Nash equilibrium of the resulting game will be $(B,B)$, which gives her at least $9$ utility. Since player $1$ can never expect to get more than $9$ utility in $(C,C)$ (even with side payments), it is always optimal for player $1$ to run this strategy. This gives $2$ utility to player $2$ and, therefore, the expected utility that player $2$ gets with commitments is strictly less than without commitments.

The issue with the previous example is that player $1$ can make the $(A,A)$ equilibrium disappear in her first move, forcing player $2$ to converge on $(B,B)$. For instance, if $\delta < 0.25$ and $(A,A)$ gave $5.5$ utility to each player instead, players can play a strategy in which they decrease by $\delta$ all utilities except the ones in $(A,A)$ and $(C,C)$ in every step of the commitment phase until $(C,C)$ becomes a Nash equilibrium. If any player does something different, they stop the commitment phase and play $(A,A)$. It is straightforward to check that such a strategy is a subgame perfect equilibrium. Since $(A,A)$ remains an equilibrium even if a player defects, the threat of playing $(A,A)$ persists and deters any player from defecting.

This analysis suggests that, for the commitment process to work as intended, it is necessary that players cannot modify the original equilibrium with a single move. In particular, we require the following notion of equilibrium.

\begin{definition}[Punishable equilibrium]\label{def:robust}
Let $\Gamma = (P,A,U)$ be a normal-form game, $\vec{\sigma}$ be a Nash equilibrium of $\Gamma$, and let $\epsilon$ and $\delta$ be two strictly positive real numbers. We say that $\vec{\sigma}$ is an $(\epsilon, \delta)$-punishable equilibrium if, whenever $d(\Gamma, \Gamma') < \delta$ for some game $\Gamma'$, there exists a Nash equilibrium $\vec{\sigma}'$ of $\Gamma'$ such that $u'_i(\vec{\sigma}') - u_i(\vec{\sigma}) < \epsilon$ for all $i \in P$. If furthermore there exists a Nash equilibrium $\vec{\sigma}'$ that also satisfies that 
$\sigma_i(A_i) = \sigma'_i(A_i)$, we say that $\vec{\sigma}$ is $(\epsilon, \delta)$-strongly punishable. Moreover, we say that $\vec{\sigma}$ is punishable (resp., strongly punishable) if, for all $\epsilon > 0$ there exists a $\delta > 0$ such that $\vec{\sigma}$ is $(\epsilon, \delta)$-punishable (resp., $(\epsilon, \delta)$-strongly punishable).
\end{definition}

Intuitively, a Nash equilibrium $\vec{\sigma}$ is punishable if, after applying small perturbations on the utilities of the game, we can always find an equilibrium $\vec{\sigma}'$ in the new game such that the utilities obtained by the players are not much higher than in $\vec{\sigma}$ (played in the original game).
We will refer to these strategy profiles $\vec{\sigma}'$ as \emph{punishment equilibria}.
Since it is generally hard to determine if a given Nash equilibrium 
is punishable, we will restrict our analysis to equilibria that satisfy a much (computationally) simpler property that implies strong punishability. We call this property \emph{non-degeneration}, and the vast majority of equilibria satisfy it. To introduce this concept and to show why it implies strong punishability, it is better to start with an example. Consider the following game for two players in which $A_1 = A_2 = \{a_1, a_2, a_3\}$ and in which the payoff matrix is as as follows.

$$\begin{array}{|c||c|c|c|}
\hline
 & a_1 & a_2 & a_3 \\
 \hhline{|=#=|=|=|}
 a_1 & (5,5) & (1, 1) & (1,0) \\
 \hline
 a_2 & (1,1) & (5, 5) & (0,1)\\
 \hline
 a_3 & (0,1) & (1,0) & (2,2) \\
 \hline
\end{array}$$

Suppose that we want to determine if there exists a Nash equilibrium $\vec{\sigma}$ such that $\sigma_1(A_1) = \sigma_2(A_2) = \{a_1, a_2\}$. 
Then, $\vec{\sigma}$ must satisfy the following system of equations:

$$
\begin{array}{rcll}
\sigma_1^1 + \sigma_1^2 & = & 1 & \mbox{(Probabilities must add up to 1.)}\\
\sigma_2^1 + \sigma_2^2  &=& 1 & \mbox{(Probabilities must add up to 1.)}\\
u_1(a_1, \sigma_2) - u_1(a_2, \sigma_2) &=& 0 & (\mbox{Player 1 is indifferent between } a_1 \mbox{ and } a_2.)\\
u_1(a_1, \sigma_2) - u_1(a_3, \sigma_2) &\ge& 0 & (\mbox{Player 1 weakly prefers } a_1 \mbox{ to } a_3.)\\
u_2(a_1, \sigma_1) - u_2(a_2, \sigma_1) &=& 0 & (\mbox{Player 2 is indifferent between } a_1 \mbox{ and } a_2.)\\
u_2(a_1, \sigma_1) - u_2(a_3, \sigma_1) &\ge& 0 & (\mbox{Player 2 weakly prefers } a_1 \mbox{ to } a_3.)\\
\end{array}
$$

This system is composed by four linear equations and two linear inequalities. In particular, the four linear equations can be written in matricial form as 
$$\begin{pmatrix}
1 & 1 & 0 & 0\\
4 & -4 & 0 & 0\\
0 & 0 & 1 & 1\\
0 & 0 & 4 & -4
\end{pmatrix}
\begin{pmatrix}
\sigma_1^1 \\
\sigma_1^2 \\
\sigma_2^1 \\
\sigma_2^2
\end{pmatrix}
 = 
 \begin{pmatrix}
1 \\
0 \\
1 \\
0
\end{pmatrix}.
$$

This system gives a unique solution in which $\sigma_1^1 = \sigma_1^2 = \sigma_2^1 = \sigma_2^2 = \frac{1}{2}$. It is easy to check that this solution also strictly satisfies the two inequalities, and therefore that it is indeed a Nash equilibrium. Suppose now that we add small perturbations to the utilities of the players. If these perturbations are small enough and we repeat the same process to find a similar Nash equilibrium, the resulting matrix will still be invertible, its unique solution will still strictly satisfy the two inequalities, and the expected utilities of the players will not be too far off from the original ones. This shows that $\vec{\sigma}$ is a strongly punishable Nash equilibrium.

Note that the two main ingredients that we need in order to argue about the strong punishability of $\vec{\sigma}$ are the fact that (a) the matrix obtained from the system of equations in order to compute $\vec{\sigma}$ is invertible, and (b) that the inequalities are \textbf{strictly} satisfied (i.e., with no equality). We can generalize this idea for any two-player game, and later for multi-player games. Consider a normal-form game $\Gamma = (P, A, U)$ for two players and a Nash equilibrium $\vec{\sigma}$.
Without loss of generality assume that $\sigma_1(A_1) = \{a_1^1, a_1^2, \ldots, a_1^{M_1}\}$ and $\sigma_2(A_2) = \{a_2^1, a_2^2, \ldots, a_2^{M_2}\}$. 
Then, we define the characteristic function of $\vec{\sigma}$ as
$$f_{\vec{\sigma}} (p_1^1, p_1^2, \ldots, p_1^{M_1}, p_2^1, p_2^2, \ldots, p_2^{M_2}) := 
\left \{
\begin{array}{l}
p_1^1 + p_1^2 + \ldots + p_1^{M_1}\\
p_2^1 + p_2^2 + \ldots + p_2^{M_2}\\
u_1(a_1^1, \vec{p}_2) - u_1(a_1^2, \vec{p}_2)\\
u_1(a_1^1, \vec{q}) - u_1(a_1^3, \vec{p}_2)\\
\vdots\\
u_1(a_1^1, \vec{p}_2) - u_1(a_1^{M_1}, \vec{p}_2)\\
u_2(a_2^1, \vec{p}_1) - u_2(a_2^2, \vec{p}_1)\\
u_2(a_2^1, \vec{p}_1) - u_2(a_2^3, \vec{p}_1)\\
\vdots\\
u_2(a_2^1, \vec{p}_1) - u_2(a_2^{M_2}, \vec{p}_1),\\
\end{array}
\right.$$

where we identify $\vec{p}_i$ as the strategy in which player $i$ plays action $1$ with probability $p_i^1$, plays action $2$ with probability $p_i^2$, etc. We also define the residual function of $\vec{\sigma}$ as 

$$r_{\vec{\sigma}} (p_1^1, p_1^2, \ldots, p_1^{M_1}, p_2^1, p_2^2, \ldots, p_2^{M_2}) := 
\left \{
\begin{array}{l}
u_1(a_1^1, \vec{p}_2) - u_1(a_1^{M_1 + 1}, \vec{p}_2)\\
u_1(a_1^1, \vec{p}_2) - u_1(a_1^{M_1 + 2}, \vec{p}_2)\\
\vdots \\
u_1(a_1^1, \vec{p}_2) - u_1(a_1^{N_1}, \vec{p}_2)\\
u_2(a_2^1, \vec{p}_1) - u_2(a_2^{M_2 + 1}, \vec{p}_1)\\
u_2(a_2^1, \vec{p}_1) - u_2(a_2^{M_2 + 2}, \vec{p}_1)\\
\vdots \\
u_2(a_2^1, \vec{p}_1) - u_2(a_2^{N_2}, \vec{p}_1)\\
\end{array}
\right.$$

As in the example, $f_{\vec{\sigma}}$ is a linear function from $[0,1]^{M_1 + M_2}$ to $\mathbb{R}^{M_1 + M_2}$.
It is straightforward to check that condition (a) from the example is equivalent to $|Df_{\vec{\sigma}}(\vec{\sigma})| \not = 0$ (where $|Df_{\vec{\sigma}}(\vec{\sigma})|$ denotes the determinant of the Jacobian of $f_{\vec{\sigma}}$ evaluated at $\vec{\sigma}$), and condition (b) is equivalent to $r_{\vec{\sigma}}(\vec{\sigma}) > 0$ (i.e., all its components are strictly positive when evaluated at $\vec{\sigma}$). This gives the following definition.

\begin{definition}[Non-Degenerate Equilibrium]\label{def:non-degenerate}
If $\Gamma$ is a normal-form game and $\vec{\sigma}$ is a Nash equilibrium of $\Gamma$, we say that $\vec{\sigma}$ is non-degenerate if $|Df_{\vec{\sigma}}(\vec{\sigma})| \not = 0$ and $r_{\vec{\sigma}}(\vec{\sigma}) > 0$.
\end{definition}

Note that we haven't specified that $\Gamma$ must be a two-player game. This is because the construction of $f_{\vec{\sigma}}$ and $r_{\vec{\sigma}}$ can be easily generalized for any number of players. Consider a normal-form game $\Gamma = (P,A, U)$ for $n$ players and a Nash equilibrium $\vec{\sigma}$ for $\Gamma$. 
Suppose that $\sigma_i(A_i) = \{a_i^1, \ldots, a_i^{M_i}\}$.
Then, 

$$f_{\vec{\sigma}} (p_1^1, \ldots, p_n^{M_n}) := 
\left \{
\begin{array}{l}
p_1^1 + p_1^2 + \ldots + p_1^{M_1}\\
p_2^1 + p_2^2 + \ldots + p_2^{M_2}\\
\vdots\\
p_n^1 + p_n^2 + \ldots + p_n^{M_n}\\
u_1(a_1^1, \vec{p}_{-1}) - u_1(a_1^2, \vec{p}_{-1})\\
u_1(a_1^1, \vec{p}_{-1}) - u_1(a_1^3, \vec{p}_{-1})\\
\vdots\\
u_n(a_n^1, \vec{p}_{-n}) - u_n(a_n^{M_n}, \vec{p}_{-n}),\\
\end{array}
\right.$$

where $\vec{p}_{-i}$ denotes the strategy profile in which each player $j \not = i$ plays $a_k^j$ with probability $p_k^j$. We can also define $r_{\vec{\sigma}}$ as  

$$r_{\vec{\sigma}} (p_1^1, \ldots, p_{M_n}^n) := 
\left \{
\begin{array}{l}
u_1(a_1^1, \vec{p}_{-1}) - u_1(a_1^{M_1 + 1}, \vec{p}_{-1})\\
u_1(a_1^1, \vec{p}_{-1}) - u_1(a_1^{M_1 + 2}, \vec{p}_{-1})\\
\vdots \\
u_n(a_n^1, \vec{p}_{-n}) - u_n(a_n^{N_n}, \vec{p}_{-n}).\\
\end{array}
\right.$$

If $n > 2$, the components of both $f_{\vec{\sigma}}$ and $r_{\vec{\sigma}}$ are no longer linear. However, perhaps surprisingly, non-degeneration implies strong punishability even in this case (see Appendix~\ref{sec:thm-strong} for the proof.)

\begin{theorem}\label{thm:strong-robust}
If $\Gamma$ is a normal-form game and $\vec{\sigma}$ is a non-degenerate Nash equilibrium of $\Gamma$, then $\vec{\sigma}$ is strongly punishable.
\end{theorem}

\section{Main Results}\label{sec:main-results}

Consider a normal-form game $\Gamma$ for $n$ players and a non-degenerate Nash equilibrium $\vec{\sigma}$ of $\Gamma$. The main result of this paper is showing that, if players have the power to sequentially perform small commitments before the game starts, then, for any utility distribution that (a) achieves maximum social welfare, and (b) Pareto improves the expected utilities of $\vec{\sigma}$, there exists a $\delta$-commitment equilibrium of $\Gamma$ that implements these utilities.

\begin{theorem}\label{thm:main}
Let $\Gamma = (P,A,U)$ be a normal-form game for $n$ players and let $\vec{\sigma}$ be a non-degenerate Nash equilibrium of $\Gamma$. Let $x_1, \ldots, x_n$ be $n$ real numbers such that $\sum_{i = 1}^n x_i = w_{\max}(\Gamma)$ and such that $x_i > u_i(\vec{\sigma})$ for all $i \in P$. Then, for some $\delta > 0$, there exists a $\delta$-commitment equilibrium $\vec{\tau}$ of $\Gamma$ such that $u_i (\vec{\tau}) = x_i$ for all $i \in P$.
\end{theorem}

If players cannot commit to payments between them and they can only burn money, it can be shown that they can still converge on all outcomes that Pareto-improve their utility (even though they may not always exist).

\begin{theorem}\label{thm:bot}
Let $\Gamma = (P,A,U)$ be a normal-form game for $n$ players, let $\vec{\sigma}$ be a Nash equilibrium of $\Gamma$, and let $\vec{a}^{PI} \in A$ be an outcome of $\Gamma$ that strictly Pareto improves $\vec{\sigma}$. Then, there exists $D > 0$ such that, for all $\delta \in (0,D)$, there exists a $(\delta, \bot)$-commitment equilibrium of $\Gamma$ in which the utility of $\vec{a}^{PI}$ is unchanged and players end up playing $\vec{a}^{PI}$ with probability $1$.
\end{theorem}

Note that in this setting, even if Theorem~\ref{thm:bot} shows that players can still Pareto improve their utility whenever there exists an outcome that does so, it may happen that all the welfare-maximizing outcomes do not Pareto improve the original equilibrium. 

\subsection{About the Computational Complexity of the Solutions}\label{sec:complexity}

The protocols that we provide to prove Theorem~\ref{thm:main} and Theorem~\ref{thm:bot} require $O(n \cdot \delta^{-1} U_{\max}A_{\max})$ commitment rounds, where $n$ is the number of players, $U_{\max}$ is the maximum utility that a player can get in the original game, and $A_{max} = \max_{i \in P}\{|A_i|\}$. Moreover, given the $\delta$ parameter, our construction is explicit on the equilibrium path, although it relies on finding a punishment equilibrium whenever a player deviates. However, this punishment equilibrium is guaranteed to exist and to have the same support as the original equilibrium $\vec{\sigma}$. Therefore, its computation is reduced to solving a system of multivariable polynomial equations. Note that, even though we are not providing an explicit bound for $\delta$, its existence is usually guaranteed through the application of Theorem~\ref{thm:strong-robust}. Thus, it is possible to obtain a concrete lower bound by backtracking its proof. However, in general, the resulting expression would be quite convoluted and would most likely give a much lower estimate than what is actually required.

The proofs of the main theorems are done in the following order. First we prove Theorem~\ref{thm:bot} in the case where $\vec{\sigma}$ has partial support (Section~\ref{sec:partial-supp}). Then, we prove it for two-player games where $\vec{\sigma}$ has full support and each player has at least three actions (Section~\ref{sec:two-players}), followed by a proof for arbitrary normal-form games (Appendix~\ref{sec:arb-games}) and for two-player games with binary actions (Appendix~\ref{sec:2pl-2act}). We then extend these techniques to prove Theorem~\ref{thm:main} as well (Appendix~\ref{sec:thm-main}). 

\section{Proof of Theorem~\ref{thm:bot}}\label{sec:thm-bot}

For simplicity, we will prove a weaker version of Theorem~\ref{thm:bot} that states that, given $\Gamma$, $\vec{\sigma}$ and $\vec{a}^{PI}$, there exists $\delta > 0$ and a $(\delta, \bot)$-commitment equilibrium of $\Gamma$ in which the utility of $\vec{a}^{PI}$ is unchanged and players end up playing $\vec{a}^{PI}$ with probability $1$. It is easy to check that the same constructions that we present in this section are also valid for all smaller values of $\delta$ as well.

As mentioned in Section~\ref{sec:intro}, the key idea for proving Theorem~\ref{thm:bot} is showing that we can always find a $\delta > 0$ and a commitment protocol $\vec{\pi}^\delta$ in $\Gamma^\delta_\bot$ in which (a) after each commitment step, for all perturbations that modify the entries of the payoff matrix by at most $\delta$,  there exists a punishment equilibrium in the resulting game that gives each player less utility in expectation than $\vec{a}^{PI}$, and (b) eventually $\vec{a}^{PI}$ becomes a Nash equilibrium. If we can construct such a protocol, we can easily extend it to a subgame perfect equilibrium $\vec{\tau}^{\vec{\pi}^\delta}$ in $\Gamma^\delta_\bot$. In this extension, each player $i$ does as follows. At each step of the commitment phase, $i$ runs $\pi^\delta$ to perform its commitments. After the commitment step, if $\vec{a}^{PI}$ is a Nash equilibrium, $i$ votes to stop the commitment phase and plays $a_i^{PI}$. Otherwise, there are two options. If no player defected during the previous commitment step, $i$ votes to continue the commitment phase and repeats this process. Otherwise, $i$ votes to stop the commitment phase and plays its part of the punishment equilibrium described in (a). If a player ever votes to stop the commitment phase before $\vec{a}^{PI}$ becomes a Nash equilibrium, $i$ plays its part of the punishment equilibrium. To check that this is indeed a Nash equilibrium, note that a player can only defect in the following ways:

\begin{itemize}
    \item By defecting from $\pi$ during the commitment phase.
    \item By voting to continue the commitment phase when it should vote to stop.
    \item By voting to stop when it should vote to continue the commitment phase.
    \item By playing a different action after the commitment phase.
\end{itemize}

If $i$ defects from $\pi$ during the commitment phase, it guarantees that the remaining players will vote to stop the commitment phase right afterwards and play their part of the punishment equilibrium, which always exists because of property (a). This ensures that $i$ gets a lower utility than in $\vec{\tau}^{\vec{\pi}^\delta}$. This also happens if $i$ votes to stop the commitment phase before $\vec{a}^{PI}$ becomes a Nash equilibrium. If $i$ votes to continue whenever it should vote to stop, all the remaining players will vote to stop anyways and $i$'s vote will not matter. To see that $i$ does not benefit from playing a different action note that, regardless of what happens during the game, the remaining players will end up playing either their part of the punishment equilibrium or their part of $\vec{a}^{PI}$. Since both are Nash equilibria when played, it is a best response for $i$ to play its part too. Altogether this shows the following.

\begin{lemma}\label{lemma:extension}
If $\pi^\delta$ is a commitment protocol in $\Gamma^\delta_\bot$ that satisfies properties (a) and (b), then it can be extended to a subgame perfect equilibrium $\vec{\tau}^{\vec{\pi}^\delta}$ of $\Gamma^\delta_\bot$.
\end{lemma}

In some cases (namely the ones with partial support), we will design commitment protocols that satisfy a property that is slightly stronger than (a). Let $L_{\vec{\sigma}}^{\vec{a}^{PI}}$ be the minimum additional utility that a player gets from playing $\vec{a}^{PI}$ instead of playing $\vec{\sigma}$. More precisely, $$L_{\vec{\sigma}}^{\vec{a}^{PI}} := \min_{i \in P} \{u_i(\vec{a}^{PI}) - u_i(\vec{\sigma})\}.$$

We define the property (a1) as follows.

\begin{itemize}
    \item [(a1)] $\vec{\sigma}$ remains a $\left(L_{\vec{\sigma}}^{\vec{a}^{PI}}, \delta\right)$-strongly punishable Nash equilibrium after each commitment step.
\end{itemize}

It is easy to check that (a1) implies (a). In fact, if a commitment protocol satisfies (a1), it means that $\vec{\sigma}$ is a Nash equilibrium throughout the whole process. Moreover, since the utilities on given outcomes can only decrease, the expected utilities of the players when playing $\vec{\sigma}$ can only decrease as well. This means that the maximum utility that a player $i$ can get is bounded by $u_i(\vec{\sigma}) + L_{\vec{\sigma}}^{\vec{a}^{PI}}$, which by definition is lower than $u_i(\vec{a}^{PI})$. This gives the following Corollary.

\begin{corollary}\label{cor:protocol}
If $\pi^\delta$ is a commitment protocol in $\Gamma^\delta_\bot$ that satisfies properties (a1) and (b), then it can be extended to a subgame perfect equilibrium $\vec{\tau}^{\vec{\pi}^\delta}$ of $\Gamma^\delta_\bot$.
\end{corollary}

To see how to construct such a protocol $\pi^\delta$ we will analyze a couple of simple examples that cover the cases in which $\vec{\sigma}$ has partial support, which means that there exist action profiles that are played with probability $0$. Some of the techniques used in these examples will be later extended in Section~\ref{sec:full-support} for arbitrary normal-form games and non-degenerate Nash equilibria.

\subsection{Cases with partial support}\label{sec:partial-supp}

We start this case analysis with a couple of simple examples that can be easily generalized. Consider a normal-form game $\Gamma$ for two players in which the action set of each player is $\{a_1, a_2, a_3, a_4\}$ and the utility matrix is as follows.

\begin{example}\label{ex:simple-case}
$$\begin{array}{|c||c|c|c|c|}
\hline
 & a_1 & a_2 & a_3 & a_4 \\
 \hhline{|=#=|=|=|=|}
 a_1 & (2,2) & (5, 2) & (2 , 5) & (6,0)\\
 \hline
 a_2 & (2,5) & (2, 2) & (5,2) & (0,0)\\
 \hline
 a_3 & (5,2) & (2,5) & (2,2) & (0,0) \\
 \hline
 a_4 & (0,6) & (0,0) & (0,0) & (4,4) \\
 \hline
\end{array}$$
\end{example}

It is easy to check that the strategy profile $\vec{\sigma}$ in which both players choose at random between $a_1, a_2$ and $a_3$ is a non-degenerate Nash equilibrium that gives $3$ utility to each player in expectation. To see this, note that the game obtained by erasing $a_4$ from the action set of both players is equivalent to rock-paper-scissors, that the utility that players get by playing $a_4$ when the other player randomizes between $a_1, a_2$ and $a_3$ is $2$ (which is strictly less than the utility they would get randomizing), and that $\Gamma$ satisfies the properties of Definition~\ref{def:non-degenerate}. Therefore, by Theorem~\ref{thm:strong-robust}, there exists a $\delta > 0$ such that $\vec{\sigma}$ is $(1, \delta)$-strongly punishable, which means that 
any change made to the payoff matrix in which utilties change at most $\delta$ results in a game where there exists a Nash equilibrium with the same support as $\vec{\sigma}$ in which no player gets more than $4$ utility.

Hence, if players want to converge on $(a_4,a_4)$ in $\Gamma^\delta_\bot$, they can make commitments following the diagram below and play $(a_4, a_4)$ afterwards.

\begin{center}
\includegraphics[scale=0.5]{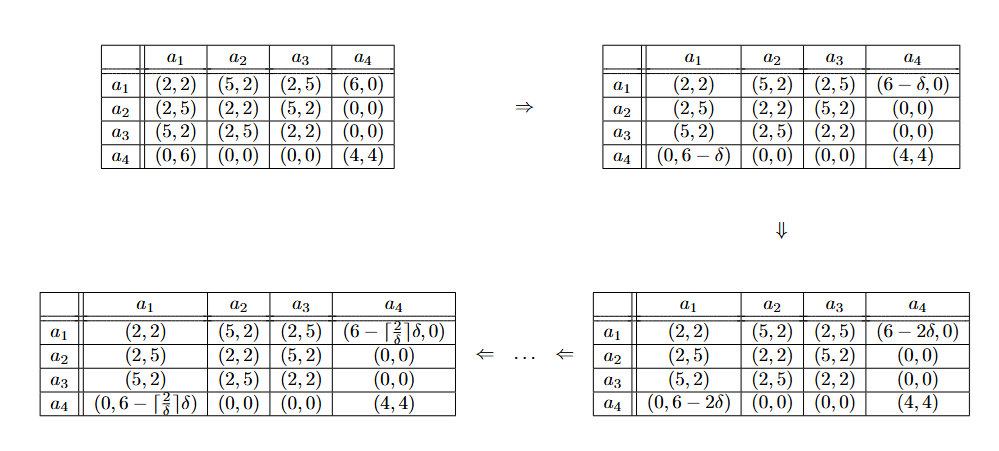}
\end{center}

Since $f_{\vec{\sigma}}$ is preserved in each commitment step and $r_{\vec{\sigma}}$ decreases, we can deduce that $\vec{\sigma}$ satisfies property (a1). Moreover, it satisfies property (b) by constructions. Therefore, by Corollary~\ref{cor:protocol}, this commitment protocol can be extended to a subgame perfect equilibrium.

The main insight used for this argument is the fact that decreasing the values outside of the support of $\vec{\sigma}$ and $\vec{a}^{PI}$ results in a game that satisfies (a1) and (b). This can be generalized to any normal-form game in which $\vec{a}^{PI}$ satisfies that $\sigma_i^{PI} = 0$ for all players $i$ (i.e., when the probability of playing $a_i^{PI}$ with $\sigma_i$ is $0$ for all $i \in P$). In this case, during each commitment step, each player $i$ should decrease by $\delta$ the utility of all outcomes in $(A \setminus \{a_i^{PI}\}) \times \{\vec{a}_{-i}^{PI}\}$ until $\vec{a}^{PI}$ becomes a Nash equilibrium. This is illustrated in the diagram below, where $-\infty$ represents a sufficiently negative value.

\begin{center}
\includegraphics[scale=0.7]{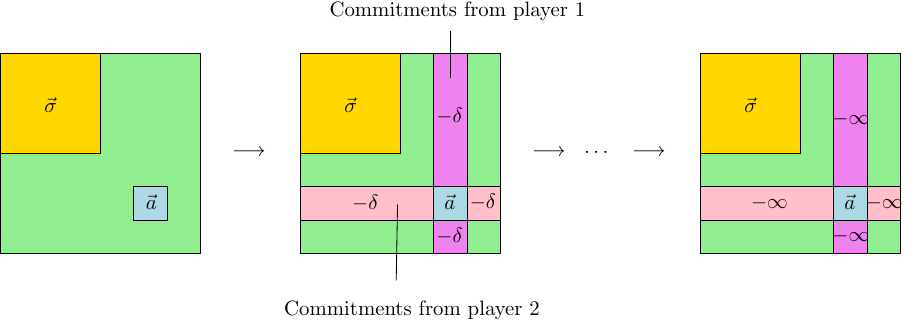}
\end{center}

A slightly more complicated case occurs when $\vec{a}^{PI}$ is not in the support of $\vec{\sigma}$ but, for some player $i$, $a_i^{PI}$ is in the support of $\sigma_i$. For example, consider a small variant of Example~\ref{ex:simple-case}.

\begin{example}\label{ex:simple-case2}
$$\begin{array}{|c||c|c|c|c|}
\hline
 & a_1 & a_2 & a_3 & a_4 \\
 \hhline{|=#=|=|=|=|}
 a_1 & (2,2) & (5, 2) & (2 , 5) & (2,2)\\
 \hline
 a_2 & (2,5) & (2, 2) & (5,2) & (0,0)\\
 \hline
 a_3 & (5,2) & (2,5) & (2,2) & (1,1) \\
 \hline
 a_4 & (2,2) & (2.5,2) & (4,4) & (0,0) \\
 \hline
\end{array}$$
\end{example}

As in Example~\ref{ex:simple-case}, the strategy profile $\vec{\sigma}$ in which each player chooses between $a_1, a_2$ and $a_3$ uniformly at random is a Nash equilibrium. However, if players want to converge on $(a_4, a_3)$ they must be careful. Taking the same approach as in Example~\ref{ex:simple-case} could potentially ``break'' $\vec{\sigma}$ since some of the commitments would lie in its support. Nevertheless,  the following commitment preserves $f_{\vec{\sigma}}$.

\begin{center}
\includegraphics[scale = 0.6]{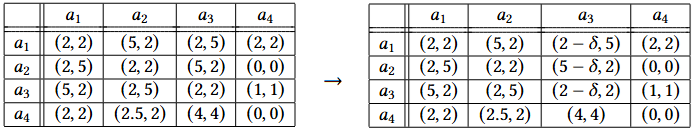}
\end{center}

Indeed,  the same approach as in Example~\ref{ex:simple-case} would work if it was not the case that $r_{\vec{\sigma}}$ increases with each iteration. In fact, if $\delta > 0.5$ we would have that $u_1(a_4, \sigma_2) > u_1(\vec{\sigma})$ and therefore that $\vec{\sigma}$ is no longer a Nash equilibrium after just one iteration. To fix this issue, the players should first start by decreasing all their utilities outside of the support of $\vec{\sigma}$ except for $(a_4, a_3)$. This way, they decrease all the components of $r_{\vec{\sigma}}$ to the point in which $r_{\vec{\sigma}}$ remains negative after running the whole procedure used in Example~\ref{ex:simple-case}. For instance, in Example~\ref{ex:simple-case2}, it suffices that player $1$ decreases by $1$ the value of $u_1(a_4, a_1), u_1(a_4, a_2)$ and $u_1(a_4, a_4)$ (this may not be done necessarily in a single step) before running the procedure used in Example~\ref{ex:simple-case}.

In short, the commitment protocol in this case goes as follows.

\begin{enumerate}
    \item First, each player $i$ decreases its utility in all outcomes outside of $\vec{a}^{PI}$ and outside of the support of $\vec{\sigma}$. It repeats this process until all such values are sufficiently low.
    \item Then, each player $i$ decreases its utility in all outcomes in $(A_i \setminus \{a_i\}) \times \vec{a}_{-i}$.
\end{enumerate}

An analogous reasoning to the one used in the analysis of Example~\ref{ex:simple-case} shows that this protocol satisfies properties (a1) and (b), and therefore it can also be extended to a subgame perfect equilibrium. In the case of two players, this is illustrated in the diagram below.

\begin{center}
\includegraphics[scale=0.7]{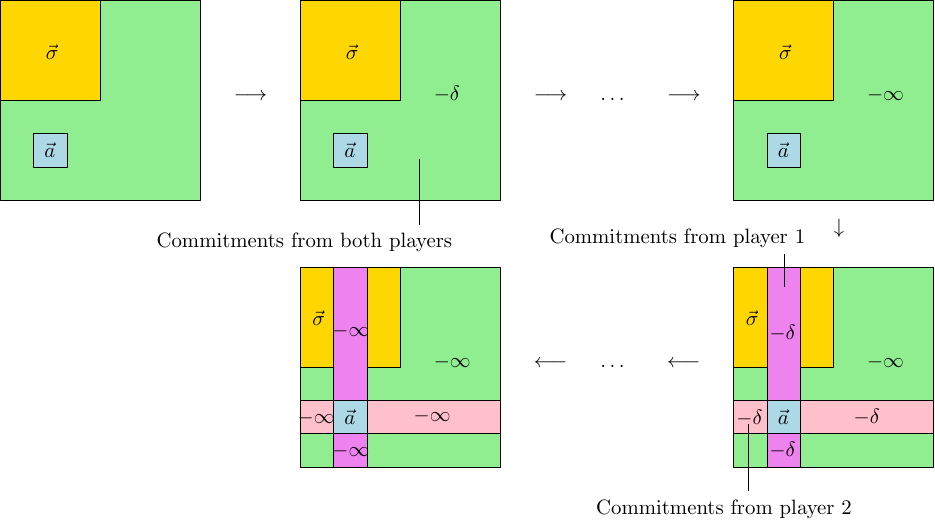}
\end{center}

The remaining case is the one in which $\vec{a}^{PI}$ is in the support of $\vec{\sigma}$, but $\vec{\sigma}$ does not have full support. Assume without loss of generality that $|A_i| \ge 2$ for all players $i$ (we can view this as ignoring all the players that cannot choose). With this assumption, there exists an action profile $\vec{a}$ such that $a_i \not = a_i^{PI}$ for all players $i$ and such that $\vec{a}$ is not in the support of $\vec{\sigma}$. Then, players can commit as follows.

\begin{enumerate}
    \item First, they decrease their utilities in $\vec{a}$ until $\vec{a}^{PI}$ strictly Pareto-dominates $\vec{a}$. 
    \item Then, they run the protocol used in the analysis of Example~\ref{ex:simple-case2} in order to make $\vec{a}$ a Nash equilibrium while satisfying properties (a1) and (b). During the second part of the protocol, each player $i$ decreases the utilities in $(A_i \setminus \{a_i^{PI}\}) \times \vec{a}_{-i}^{PI}$ until $u_i(\vec{a}) > u_i(a'_i, \vec{a}_{-i}) + \delta$ for all players $i$ and all actions $a'_i \in A_i$.
    \item After terminating the protocol above, they run again the protocol in order to make $\vec{a}^{PI}$ a Nash equilibrium, but taking $\vec{a}$ as input instead of $\vec{\sigma}$ (note that $\vec{a}$ is a Nash equilibrium at this point.)
\end{enumerate}

Note that this protocol does not satisfy property (a1). However, it still satisfies properties (a) and (b). It clearly satisfies (b) by construction. To see that it satisfies (a), note that players can use $\vec{\sigma}$ as a punishment equilibrium during the first two steps, and they can use $\vec{a}$ during the third one. This proves Theorem~\ref{thm:bot} in the cases in which $\vec{\sigma}$ does not have full support.

\subsection{Cases with Full Support}\label{sec:full-support}

The case in which $\vec{\sigma}$ has full support is considerably harder than the previous cases. The main issue is that players are forced to modify $f_{\vec{\sigma}}$ in each commitment step and, under these circumstances, it is difficult to design a protocol that satisfies (a1). In order to illustrate the main obstacles, suppose that $\vec{\sigma}$ is a $(L_{\vec{\sigma}}^{\vec{a}^{PI}}, \delta)$-punishable Nash equilibrium. If we wanted to design a commitment protocol $$\Gamma = \Gamma_0 \rightarrow \Gamma_1 \rightarrow \Gamma_2 \rightarrow \ldots \rightarrow \Gamma_k$$ that satisfies (a1), it is necessary that

\begin{itemize}
    \item [(P1)] $d(\Gamma_i, \Gamma_{i+1}) < \delta$ for all $i < k$ (using the metric defined in Section~\ref{sec:defs}).
    \item [(P2)] $u_i^j(\vec{a}) \ge u_i^{j+1}(\vec{a})$ for all $i \in P$ and all $j < k$, where $u_i^j$ denotes the utility of player $i$ in $\Gamma_j$.
    \item [(P3)] $u_i(\vec{a}^{PI}) = u_i^k(\vec{a}^{PI})$ for all $i \in P$. 
    \item [(P4)] $\vec{\sigma}$ is a non-degenerate Nash equilibrium of $\Gamma_i$ for all $i \le k$.
    \item [(P5)] $\vec{a}^{PI}$ is a Nash equilibrium of $\Gamma_k$.
\end{itemize}

Designing such a commitment protocol does not seem that complicated since it basically reduces to finding a way to decrease the coefficients of $f_{\vec{\sigma}}$ outside of $\vec{a}^{PI}$ by at most $\delta$, in each step such that $f_{\vec{\sigma}}(\vec{\sigma}) = \vec{0}$ and $|Df_{\vec{\sigma}}(\vec{\sigma})| \not = 0$ throughout the whole process. Unfortunately, this is not enough for our purposes. In fact, the non-degeneration property guarantees that, in each game $\Gamma_i$, there exists a value $\delta_i > 0$ such that $\vec{\sigma}$ is $(L_{\vec{\sigma}}^{\vec{a}^{PI}}, \delta_i)$-strongly punishable. However, $\vec{\sigma}$ may not necessarily be $(L_{\vec{\sigma}}^{\vec{a}^{PI}}, \delta)$-strongly punishable. We can try to fix the protocol by setting $\delta$ to the minimum of all such $\delta_i$, but then some of the payments involved in the process may become too large.

They key idea to circumvent this issue is making the distance between $\Gamma_i$ and $\Gamma_{i+1}$ ``infinitesimally small''. This means constructing a path that goes from the original game $\Gamma = (P, A, U)$ to a new game $\Gamma' = (P, A, U')$ in which the following conditions are satisfied:

\begin{itemize}
    \item [(P1')] $\gamma$ is continuous (using the distance defined in Section~\ref{sec:defs} as the metric) and finite.
    \item [(P2')] The utility of the players (weakly) decreases through $\gamma$.
    \item [(P3')] $u_i(\vec{a}^{PI}) = u'_i(\vec{a}^{PI})$ for all $i \in P$ (note that this, combined with the property above, implies that the utilities in outcome $\vec{a}^{PI}$ are constant in $\gamma$.)
    \item [(P4')] For each game $\Gamma'' \in \gamma$, $\vec{\sigma}$ is a non-degenerate Nash equilibrium of $\Gamma''$. 
    \item [(P5')] $\vec{a}^{PI}$ is a Nash equilibrium of $\Gamma'$.
\end{itemize}

The next proposition shows that these conditions are sufficient for our purposes.

\begin{proposition}\label{prop:curve}
If there exists a path $\gamma$ that satisfies (P1'), (P2'), (P3'), (P4') and (P5'), then, for some $\delta > 0$, it can be extended to a subgame perfect Nash equilibrium in $\Gamma_\bot^\delta$ in which players end up playing $\vec{a}^{PI}$. 
\end{proposition}

Before proving Proposition~\ref{prop:curve} we need the following compactness lemma, which is an analogue of the Heine-Cantor theorem for normal-form games.

\begin{lemma}\label{lemma:compact}
Let $X_n^A$ be the metric space of normal-form games for $n$ players with action set $A$, where the distance between two elements is the one given in Section~\ref{sec:defs}. Let $\vec{\sigma}$ be a strategy profile over $A$ and let $C \subseteq X_n^A$ be a compact set such that $\vec{\sigma}$ is a non-degenerate Nash equilibrium for all $\Gamma \in C$. Then, for all $\epsilon > 0$, there exists $\delta > 0$ such that $\vec{\sigma}$ is $(\epsilon, \delta)$-strongly punishable in all games $\Gamma \in C$. 
\end{lemma}

\begin{proof}
Let $\Gamma = (P,A,U)$ be any game in $C$ and let $\epsilon > 0$ be any positive real number. Since $\vec{\sigma}$ is a non-degenerate Nash equilibrium of $\Gamma$, by Theorem~\ref{thm:strong-robust} there exists $\delta > 0$ such that $\vec{\sigma}$ is $(\epsilon/2, \delta)$-strongly punishable. Let $\Gamma' = (P,A,U')$ be any game in $C$ at distance at most $\min(\epsilon, \delta)/2$ from $\Gamma$. Since $d(\Gamma, \Gamma') < \epsilon/2$, we have that $|u'(\vec{\sigma}) - u(\vec{\sigma})| < \epsilon/2$. Moreover, we have that all normal-form games $\Gamma''$ at distance at most $\delta/2$ from $\Gamma'$ are at most at distance $\delta$ from $\Gamma$, and therefore satisfy that there always exists a Nash equilibrium $\vec{\sigma}''$ with the same support as $\vec{\sigma}$ such that $u''_i(\vec{\sigma}'') - u_i(\vec{\sigma}) < \epsilon/2$ for all $i \in P$. By the triangular inequality, this means that $u''_i(\vec{\sigma}'') - u'_i(\vec{\sigma}) < \epsilon$ for all $i \in P$, and therefore that all games $\Gamma'$ in a ball of radius $\min(\epsilon, \delta)/2$ centered at $\Gamma$ satisfy that $\vec{\sigma}$ is a $(\delta/2, \epsilon)$-strongly punishable Nash equilibrium on those games. More generally, this means that, for each game $\Gamma \in C$, there exists an open set $B^\Gamma$ and a positive real number $\delta^\Gamma > 0$ such that $\vec{\sigma}$ is $(\epsilon, \delta^\Gamma)$-punishable on all $\Gamma' \in B^{\Gamma}$. Since $C$ is a compact set and $\{B_\Gamma\}_{\Gamma \in C}$ is an open cover of $C$, there exists a finite sub-cover $\{B_{\Gamma_1}, B_{\Gamma_2}, \ldots, B_{\Gamma_N}\}$. Let $\delta := \min_{i \le N} \{\delta_{\Gamma_i}\}$, we show next that $\vec{\sigma}$ is an $(\epsilon, \delta)$-strongly punishable equilibrium in all games $\Gamma \in C$. Since $\{B_{\Gamma_1}, B_{\Gamma_2}, \ldots, B_{\Gamma_N}\}$ is a cover of $C$, $\Gamma$ must be in $B_{\Gamma_k}$ for some $k \in \{1,2,\ldots,N\}$. Therefore, by definition, $\vec{\sigma}$ is an $(\epsilon, \delta_{\Gamma_k})$-strongly punishable equilibrium, which implies that it also is $(\epsilon, \delta)$-strongly punishable since $\delta \le \delta_{\Gamma_k}$.
\end{proof}

With this we can prove Proposition~\ref{prop:curve}.

\begin{proof}[Proof of Proposition~\ref{prop:curve}]

Suppose that such a path exists.
By property (P4') and Theorem~\ref{thm:strong-robust}
we have that, for every game $\Gamma'' \in \gamma$, there exists $\delta_{\Gamma''} > 0$ such that $\vec{\sigma}$ is a $(L_{\vec{\sigma}}^{a^{PI}}, \delta_{\Gamma''})$-strongly punishable Nash equilibrium of $\Gamma''$. By property (P1') we have that $\gamma$ is closed and bounded, and hence, by the Heine-Borel theorem, it is also a compact set. Therefore, by Lemma~\ref{lemma:compact}, there exists $\delta > 0$ such that $\vec{\sigma}$ is a $(L_{\vec{\sigma}}^{a^{PI}}, \delta)$-strongly punishable equilibrium in all $\Gamma'' \in \gamma$. Note that this means that, in any game $\Gamma'''$ at distance at most $\delta$ from some game $\Gamma'' \in \gamma$, there exists a Nash equilibrium $\vec{\sigma}'$ in which each player gets at most $L_{\vec{\sigma}}^{a^{PI}}$ additional utility than the one they would get by playing $\vec{\sigma}$ in $\Gamma''$. By property (P2'), we have that this utility is bounded by $u_i(\vec{\sigma}) + L_{\vec{\sigma}}^{a^{PI}}$, which by definition is at most $u'_i(\vec{\sigma})$.

This property ensures that the following strategy is a subgame perfect equilibrium in $\Gamma^\delta_\bot$. During each step of the commitment phase, each player commits to pay in such a way that the resulting game remains in the path $\gamma$. Intuitively, we can view this as if players perform small leaps over $\gamma$ until they reach $\Gamma'$. If a player ever defects from the proposed strategy, all players should immediately vote to terminate the commitment phase and play the punishment equilibrium. If no one defects and they eventually reach $\Gamma'$, they play $\vec{a}^{PI}$. It is easy to check that this strategy profile is indeed a subgame perfect equilibrium. If a player $i$ defects from the main strategy, it guarantees that the remaining players will vote to terminate and that the remaining players will play $\vec{\sigma}'$. In this case, it is also a best response for $i$ to vote to terminate (it doesn't change the outcome) and to play $\sigma'_i$ since it is a Nash equilibrium. However, as shown in the previous paragraph, doing this guarantees her at most $u'_i(\vec{\sigma})$ utility which, by property (P3'), is less than the utility that $i$ would get if no one defects. Moreover, if no player defects, playing $a_i^{PI}$ is also a best response since, by property (P5'), $\vec{a}^{PI}$ is a Nash equilibrium of $\Gamma'$.

\begin{center}
\includegraphics[scale=0.7]{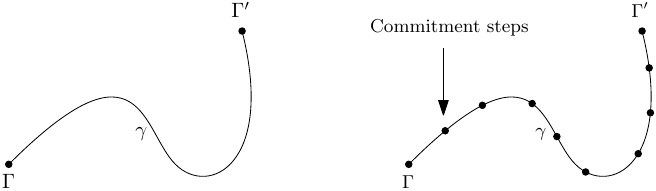}
\end{center}

\end{proof}

The argument above shows that Theorem~\ref{thm:bot} reduces to construct a path $\gamma$ that satisfies the desired conditions. In the remaining part of this section we will show how to construct such a path by concatenating linear segments. 
In order to provide more intuition, We will begin our analysis with two-player games and later generalize it for arbitrary games with any players.

\subsubsection{Two-Player Games}\label{sec:two-players}

Suppose that $\Gamma = (P,A,U)$ is a two-player game and $\vec{\sigma}$ is a non-degenerate Nash equilibrium of $\Gamma$ with full support. Assume without loss of generality that
the target outcome is $(a_1^1, a_2^1)$. Since $\Gamma$ is a two-player game, we have that, after rearranging some of its rows, $Df_{\vec{\sigma}}$ can be written as

$$Df_{\vec{\sigma}} = 
 \left(
    \begin{array}{c|c}
      X_1 & 0_{N \times N}\\
      \hline
      0_{M \times M} & X_2
    \end{array}
    \right)
$$

with 

$$X_1 = 
\begin{pmatrix}
1 & 1 & \ldots & 1 \\
    u_1(a_1^1, a_2^1) - u_1(a_1^2, a_2^1) & u_1(a_1^1, a_2^2) - u_1(a_1^2, a_2^2) &  \ldots & u_1(a_1^1, a_2^{N_2}) - u_1(a_1^2, a_2^{N_2}) \\
    u_1(a_1^1, a_2^1) - u_1(a_1^3, a_2^1) & u_1(a_1^1, a_2^2) - u_1(a_1^3, a_2^2) & \ldots & u_1(a_1^1, a_2^{N_2}) - u_1(a_1^3, a_2^{N_2}) \\
    \vdots &  \vdots & \ddots & \vdots\\
    u_1(a_1^1, a_2^1) - u_1(a_1^{N_1}, a_2^1) & u_1(a_1^1, a_2^2) - u_1(a_1^{N_1}, a_2^2) & \ldots & u_1(a_1^1, a_2^{N_2}) - u_1(a_1^{N_1}, a_2^{N_2}) \\
\end{pmatrix}
$$

and 

$$X_2 = 
\begin{pmatrix}
1 & 1 & \ldots & 1 \\
    u_2(a_1^1, a_2^1) - u_1(a_1^1, a_2^2) & u_2(a_1^2, a_2^1) - u_2(a_1^2, a_2^2) &  \ldots & u_2(a_1^{N_1}, a_2^1) - u_2(a_1^{N_1}, a_2^2) \\
    u_2(a_1^1, a_2^1) - u_2(a_1^1, a_2^3) & u_2(a_1^2, a_2^1) - u_2(a_1^2, a_2^3) & \ldots & u_2(a_1^{N_1}, a_2^1) - u_2(a_1^{N_1}, a_2^3) \\
    \vdots &  \vdots & \ddots & \vdots\\
    u_2(a_1^1, a_2^1) - u_2(a_1^1, a_2^{N_2}) & u_2(a_1^2, a_2^1) - u_2(a_1^2, a_2^{N_2}) & \ldots & u_2(a_1^{N_1}, a_2^1) - u_2(a_1^{N_1}, a_2^{N_2}) \\
\end{pmatrix}.
$$

This means that $|Df_{\vec{\sigma}}| \not = 0$ if and only if
\begin{itemize}
    \item $N = M$.
    \item $X_1$ and $X_2$ are $N \times N$ square matrices such that $|X_1| \not = 0$ and $|X_2| \not = 0$. 
\end{itemize}

The high level idea used for the construction of the path $\gamma$ is the fact that adding a multiple of a row to another row does not change the determinant of the matrix. Following this intuition, our aim is to find a commitment protocol for the players such that each commitment induces a row operation on $X_1$ or $X_2$ such that
\begin{itemize}
\item The row operation consists of adding a multiple of a row to a different row.
    \item The first row (i.e., the row with all entries equal to $1$) is not involved in the operation.
    \item Neither player modifies the utility of $(a_1^1, a_2^1)$.
    \item After a certain number of operations, the lowest entry in the first column of $X_1$ and $X_2$ is increased by a constant amount.
\end{itemize}

Note that, since the equation that corresponds to the first row - i.e., the one that states that the probabilities should add up to $1$ - is the only one that has an independent term, any row operation that does not involve the first row preserves the solution of the linear system. Therefore, the second condition guarantees that $\vec{\sigma}$ remains a Nash equilibrium after each commitment step. The third condition ensures that the utility of both players in outcome $(a_1^1, a_2^1)$ is preserved, and the fourth condition guarantees that, eventually, all entries of the first column of $X_1$ and $X_2$ will be positive, which implies that $a_1$ is eventually becomes a best response to $a_1$ for both players, and thus that $(a_1^1, a_2^1)$ eventually becomes a Nash equilibrium. Thus, if players can manage to commit in such a way that each commitment step performs a row operation on $X_1$ or $X_2$ that satisfies the properties above, we can use these commitments in succession until $(a_1^1, a_2^1)$ becomes a Nash equilibrium in the resulting game. The piece-wise linear path formed by linearly extending these row operations forms the desired path $\gamma$. More precisely, if a player's commitment performs a row operation on $X_1$ or $X_2$ that consists of replacing row $r$ by row $r + c \cdot r'$, where $r'$ is another row of the matrix, then we can connect the original and the resulting matrix with the linear segment defined by $r + \lambda \cdot c \cdot r'$, with $\lambda \in [0,1]$. The first property guarantees that the determinants of $X_1$ and $X_2$ are preserved. In particular, the determinant of $Df_{\vec{\sigma}}$ remains non-zero. 

In order to perform these row operations, consider the following two types of commitments:
\begin{itemize}
    \item $P(i,\vec{a},x)$: Player $i$ decreases her utility by $x$ in outcome $\vec{a}$.
    \item $M(i, \vec{a}, x)$: Player $i$ decreases her utility by $x$ in all outcomes of the form $(a', \vec{a}_{-i})$ with $a' \in A_i \setminus \{a_i\}$.
\end{itemize}

By construction, if $\vec{a} := (a_1^j, a_2^k)$ and $j > 1$, $P(1, \vec{a}, x)$ increases the value of the $(j,k)$-th entry of $X_1$ by $x$ and $M(1, \vec{a}, x)$ decreases its value by $x$ instead (for player $2$, the entry that is modified is $(k,j)$ instead). These two operations are represented in the diagram below.

\begin{center}
\includegraphics[scale=0.7]{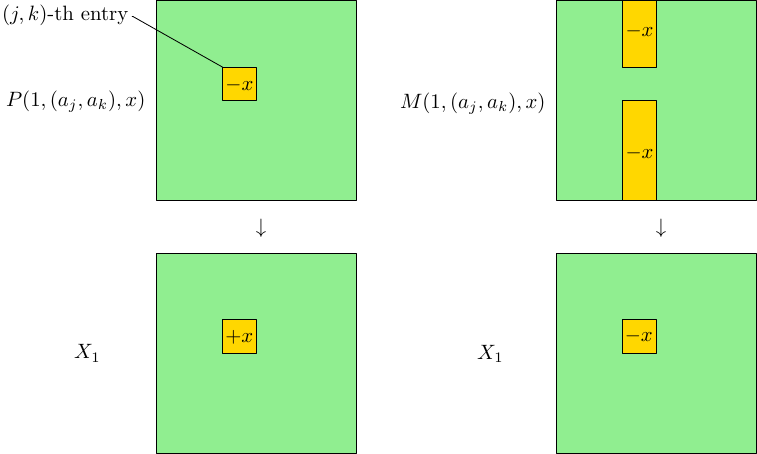}
\end{center}

Combining these types of commitments, players can perform any type of row operation on $X_1$ and $X_2$. In particular, if the number of available actions $N$ satisfies $N > 2$, each player $i$ can run the following protocol.

\begin{enumerate}
    \item \textbf{Step 1:} Check all rows of $X_i$, if there is any row except the first one in which the first entry is strictly positive, add positive multiples of this row to all other rows of $X_i$ except the first one. Repeat this process until all the entries of the first column of $X_i$ are positive.
    \item \textbf{Step 2:} If there is no row in $X_i$ such that the first entry is strictly positive, check if there is one in which the first entry is strictly negative. If it exists, add a negative multiple of this row to all other rows except the first one. Repeat this process until the first entry of all other rows is positive. Then, run Step 1 using any other row as pivot.
\end{enumerate}

Clearly, if all players follow the protocol above, the first columns of $X_1$ and $X_2$ will eventually be non-negative, and therefore $(a_1^1, a_2^1)$ will eventually become a Nash equilibrium of the resulting game. Note that the protocol can be performed in $\Gamma^\delta_\bot$ since players can always choose to add multiples of the rows that are small enough for their coordinates to be bounded by $\delta$ in absolute value. Also note that since we are always adding rows that have a positive first entry, we are not modifying the players' utilities in $(a_1^1, a_2^1)$, as desired. Altogether, this proves Theorem~\ref{thm:bot} in the case of two players and $N > 2$. The case of two players and binary actions is quite special and must be treated separately (see Appendix~\ref{sec:2pl-2act}). The proof for the case of three or more players uses similar techniques and can be found in Appendix~\ref{sec:arb-games}. The high-level idea is that players can commit in a way that they modify the components of $f_{\vec{\sigma}}$ while preserving the image of $\vec{\sigma}$ and the gradient of that component (and therefore the non-degeneration property). Thus, they can slowly converge to the desired outcome by applying these types of transformations. To prove Theorem~\ref{thm:main}, players first commit in a way that results in an outcome that gives all players the desired utility profile, and then they use the techniques of the proof of Theorem~\ref{thm:bot} to make it a Nash equilibrium of the resulting game (see Appendix~\ref{sec:thm-main} for a full proof.)

    \section{Conclusion}\label{sec:conclusion}

Given a normal-form game, we showed that giving players the chance to sequentially commit to small side payments allows them to implement all welfare-maximizing utility profiles that strictly Pareto improve any Nash equilibria of the original game. Moreover, if there already exists an outcome that strictly Pareto improves any of these Nash equilibria, that outcome can be implemented via \emph{money burning}. This aligns with Coase's theorem and comes in contrast to Jackson and Wilkie's analysis, who show that this is not always the case for a single round of unbounded commitments. Two natural questions follow from these results.

First, finding out if these positive results can be extended to a Bayesian setting where players have private information. Most of the techniques used in the construction of the commitment protocol make use of the fact that players have full information about the payoffs of the game. Thus, it would seem that the analysis of the Bayesian setting would require a different approach. 

Second, constructing a more realistic lower bound for $\delta$, even if it holds only for specific classes of games. In Section~\ref{sec:complexity} we provide an upper bound on the number of commitment rounds used in the constructions for Theorems~\ref{thm:main} and Theorem~\ref{thm:bot} as a function of the parameter $\delta$. However, we noted that the lower bound for $\delta$ that we would get with a generic estimate obtained by backtracking the proof of Theorem~\ref{thm:strong-robust} would be in general significantly lower than the actual lower bound for a vast majority of games.

\printbibliography

\appendix

\section{Proof of Theorem~\ref{thm:bot} for Arbitrary Games}\label{sec:arb-games}

The case with several players works differently than the case with only two players presented in Section~\ref{sec:two-players}. The main difference is that $f_{\vec{\sigma}}$ is no longer a linear function, and thus some of the arguments used in the previous section do not apply anymore. In particular, doing a $P(i, \vec{a}, x)$ or a $M(i, \vec{a},  x)$ commitment does not increase or decrease the value of a single entry of $Df_{\vec{\sigma}}$ by $x$. In order to see why, consider the following example.

\begin{example}\label{ex:multivariable}
Consider a game $\Gamma = (P, A , U)$ with $|P| = 3$, $A_1 = A_2 = A_3 = \{0,1\}$ and such that there exists a Nash equilibrium $\vec{\sigma}$ with full support. 
\end{example}
In this example, if we use the notation $d_i(\vec{a}_{-i}) := u_i(0, \vec{a}_{-i}) - u_i(1, \vec{a}_{-i})$, we have that 
$$
f_{\vec{\sigma}} (p_0^1, p_1^1, p_0^2, p_1^2, p_0^3, p_1^3) := 
\left \{
\begin{array}{l}
p_0^1 + p_1^1\\
p_0^2 + p_1^2\\
p_0^3 + p_1^3\\
d_1(0,0) \cdot  p_0^2p_0^3 + d_1(0,1) \cdot  p_0^2p_1^3 + d_1(1,0) \cdot  p_1^2p_0^3 + d_1(1,1) \cdot  p_1^2p_1^3\\
d_2(0,0) \cdot  p_0^1p_0^3 + d_2(0,1) \cdot  p_0^1p_1^3 + d_2(1,0) \cdot  p_1^1p_0^3 + d_2(1,1) \cdot  p_1^1p_1^3\\
d_3(0,0) \cdot  p_0^1p_0^2 + d_3(0,1) \cdot  p_0^1p_1^2 + d_3(1,0) \cdot  p_1^1p_0^2 + d_3(1,1) \cdot  p_1^1p_1^2,\\
\end{array}
\right.
$$

where $p_j^i$ is a variable that denotes the probability that player $i$ plays action $j$. To see how commitments affect $Df_{\vec{\sigma}}$, note that if player $1$ performs a $P(1, (1,0,0), x)$ commitment, the only change in $f_{\vec{\sigma}}$ is that $d_1(0,0)$ increases by $x$. This induces an increment of $(0,0, x \cdot p_0^3, 0, x \cdot p_0^2, 0)$ in the fourth row of $Df_{\vec{\sigma}}$. More generally, it can be checked that, in this example, the $P(i, \vec{a}, x)$ and $M(i, \vec{a}, x)$ commitments modify exactly two values of a single row of $Df_{\vec{\sigma}}$.

Even though this property can make us believe that the desired protocol is much more difficult to construct with respect to the case of two players, it allows us to take the following approach. As in the case of two players, we construct a piece-wise linear path. However, each piece of the path, instead of inducing a row operation on $Df_{\vec{\sigma}}$ (as in the case of two players), it will not modify $Df_{\vec{\sigma}}$ at all. To see how this works, we will introduce this technique with the example above.

However, before going into the details, it will be useful to introduce some notation. First, denote by $f_{\vec{\sigma}}^j$ the $j$th component of $f_{\vec{\sigma}}$. For instance, in Example~\ref{ex:multivariable}, we have that $$f_{\vec{\sigma}}^4 = d_1(0,0) \cdot  p_0^2p_0^3 + d_1(0,1) \cdot  p_0^2p_1^3 + d_1(1,0) \cdot  p_1^2p_0^3 + d_1(1,1) \cdot  p_1^2p_1^3.$$

Using our usual notation where  $A_1 = (a_1^1, a_1^2, a_1^3, \ldots, a_1^{N_1})$, $A_2 = (a_2^1, a_2^2, a_2^3, \ldots, a_2^{N_2})$, $\ldots$, $A_n = (a_n^1, a_n^2, a_n^3, \ldots, a_n^{N_n})$, there exists a natural lexicographical ordering of action profiles. It can be defined recursively as follows:
\begin{itemize}
    \item $a_i^j \prec a_i^k \Longleftrightarrow j < k$
    \item $(a_{i_1}^{j_1}, a_{i_2}^{j_2}, \ldots, a_{i_\ell}^{j_\ell}) \prec (a_{i_1}^{k_1}, a_{i_2}^{k_2}, \ldots, a_{i_\ell}^{k_\ell})$ if and only if $j_1 < k_1$ or $j_1 = k_1$ and $(a_{i_2}^{j_2}, \ldots, a_{i_\ell}^{j_\ell}) \prec (a_{i_2}^{k_2}, \ldots, a_{i_\ell}^{k_\ell})$.
\end{itemize}

This also induces a total order between the coefficients of $f_{\vec{\sigma}}^j$ for $j > n$. All these  coefficients are of the form $(u_i(a_i^1, \vec{a}_{-i}) - u_i(a_i^k, \vec{a}_{-i}))$ for some fixed player $i$ and some fixed action index $k$. Then we can define the following order between coefficients: 

$$(u_i(a_i^1, \vec{a}_{-i}) - u_i(a_i^k, \vec{a}_{-i})) \prec (u_i(a_i^1, \vec{a}'_{-i}) - u_i(a_i^k, \vec{a}'_{-i})) \Longleftrightarrow \vec{a}_{-i} \prec \vec{a}'_{-i}.$$

In Example~\ref{ex:multivariable}, this order is given by the following relation (note that the first action index is $0$ instead of $1$): $$d_i(\vec{a}_{-i}) \prec d_i(\vec{a}'_{-i}) \Longleftrightarrow \vec{a}_{-i} \prec \vec{a}_{-i}'.$$

In particular, in $f_{\vec{\sigma}}^4$, we would have the following order:

$$d_1(0,0) \prec d_1(0,1) \prec d_1(1,0) \prec d_1(1,1).$$

For simplicity, we will assume that the coefficients of $f^i_{\vec{\sigma}}$ are always ordered following this relation. Given a positive integer $k$ and an array of real values $\vec{x} = (x_1, x_2, \ldots)$ with the appropriate dimension, we denote by $R(j, \vec{x})$ the commitment by which a player modifies the coefficients of $f_{\vec{\sigma}}^j$ by $x_1$, $x_2, \ldots$, respectively. More precisely, by performing a $R(j, \vec{x})$ commitment, we increase the first coefficient of $f_{\vec{\sigma}}^j$ by $x_1$, the second one by $x_2$, etc. Note that this commitment can always be performed by adding together commitments of the form $P(i, \vec{a}', x)$ and $M(i, \vec{a}', x)$. For instance, in Example~\ref{ex:multivariable}, if $x_1, x_2, x_3$ and $x_4$ are positive real values, we have that $$
\begin{array}{lll}
R(4, (x_1, -x_2, -x_3, x_4)) & = &  P(1, (1,0,0), x_1) \\
& \oplus & M(1, (1,0,1), x_2)\\
& \oplus & M(1, (1,1,0), x_3) \\
& \oplus & P(1, (1,1,1), x_4),
\end{array}$$

where $\oplus$ indicates the concatenation of commitments (i.e., performing those commitments at the same time). We also denote by $f^j_{\vec{\sigma}} \oplus R(j, \vec{x})$ the resulting function obtained from performing an $R(j, \vec{x})$ commitment on $f^j_{\vec{\sigma}}$.
In Example~\ref{ex:multivariable}, if we set $x_1 = \lambda\cdot p_1^2p_1^3$, $x_2 = \lambda \cdot p_1^2p_0^3$, $x_3 = \lambda \cdot p_0^2p_1^3$ and $x_4 = \lambda \cdot p_0^2p_0^3$, we have that 
\begin{equation}\label{eq:commitment}
\begin{array}{lll}
\nabla (f^4_{\vec{\sigma}} \oplus R(4, (x_1, -x_2, -x_3, x_4)))(\vec{p}) - \nabla f^4_{\vec{\sigma}}(\vec{p}) & = & \lambda \cdot p_2^1p_3^1 \cdot \nabla(p_2^0 p_3^0) \\
& - & \lambda\cdot p_2^1p_3^0 \cdot \nabla(p_2^0 p_3^1) \\
& - & \lambda \cdot p_2^0p_3^1 \cdot \nabla (p_2^1 p_3^0) \\
& + & \lambda \cdot p_2^0p_3^0 \cdot \nabla (p_2^1 p_3^1).
\end{array}
\end{equation}

The factor by which $\lambda$ is multiplied can be expanded as

$$
\begin{array}{lllllllll}
& (& 0, & 0,& p_2^1p_3^1p_3^0, & 0, & p_2^1p_3^1p_2^0, & 0&)\\
+ & (&0,& 0,& -p_2^1p_3^0p_3^1, & 0, & 0, & -p_2^1p_3^0p_2^0&)\\
+ & (&0,& 0,& 0,& -p_2^0p_3^1p_3^0,&  -p_2^1p_3^1p_2^1,&  0&)\\
+ & (&0,& 0,& 0,& p_2^0p_3^0p_3^1,& 0,& p_2^0p_3^0p_2^1&),\\
\end{array}
$$

which shows that 

$$\nabla (f^4_{\vec{\sigma}} \oplus R(4, (x_1, -x_2, -x_3, x_4)))(\vec{p}) - \nabla f^4_{\vec{\sigma}}(\vec{p}) = \vec{0}$$

for all $\lambda \in \mathbb{R}$. Moreover, we also have the following. 

\begin{equation}\label{eq:commitment2}
\begin{array}{lll}
 (f^4_{\vec{\sigma}} \oplus R(4, (x_1, -x_2, -x_3, x_4)))(\vec{p}) - f^4_{\vec{\sigma}}(\vec{p}) & = & \lambda \cdot p_2^1p_3^1 \cdot (p_2^0 p_3^0) \\
& - & \lambda\cdot p_2^1p_3^0 \cdot (p_2^0 p_3^1) \\
& - & \lambda \cdot p_2^0p_3^1 \cdot (p_2^1 p_3^0) \\
& + & \lambda \cdot p_2^0p_3^0 \cdot (p_2^1 p_3^1).
\end{array}
\end{equation}

which is the same as equation~\ref{eq:commitment} but without the gradient operator. It can be easily checked that equation~\ref{eq:commitment2} is also equal to $0$. This means that player $1$ can modify the values of the coefficients of $f_{\vec{\sigma}}^4$ following the linear path defined by $R(4, \lambda\cdot (p_2^1p_3^1, p_2^1p_3^0, p_2^0p_3^1, p_2^0p_3^0))$ without modifying $Df_{\vec{\sigma}}(\vec{p})$ or $f_{\vec{\sigma}}(\vec{p})$. In particular, if we take $\vec{p} = \vec{\sigma}$, we have that player $1$ can modify the coefficients of $f_{\vec{\sigma}}^4$ following this path while simultaneously satisfying that $Df_{\vec{\sigma}}(\vec{\sigma})$ remains constant and that $\vec{\sigma}$ remains a Nash equilibrium in the resulting game. Suppose that $\vec{a}^{PI} = (0,0,0)$. If we choose positive values of $\lambda$, we can also guarantee that the utility of player $1$ in $\vec{a}^{PI}$ remains untouched and that eventually $0$ becomes a best response when the remaining players play $(0,0)$. Having players $2$ and $3$ follow similar commitment patterns gives us the desired path $\gamma$.

We next extend this argument to arbitrary games. Let $\Gamma = (P, A, U)$ with $P = \{1,2,\ldots, n\}$, $n \ge 3$, and $A_i = \{a_i^1, a_i^2, \ldots, a_i^{N_i}\}$ for each $i \in P$. Assume without loss of generality that $|A_i| \ge 2$ for all $i \in P$ and that $\vec{a}^{PI} = (a_1^1, a_2^1, \ldots, a_n^1)$. Our aim is to show that, for each component $f_{\vec{\sigma}}^j$ with $j > n$ (i.e., a component that is not of the form $p_i^1 + p_i^2 + \ldots + p_i^{N_i}$), there always exists an array $\vec{x}^j$ such that
\begin{itemize}
\item [(Q1)] $(f_{\vec{\sigma}}^j \oplus R(j, \vec{x}^j))(\vec{\sigma}) - f_{\vec{\sigma}}^j(\vec{\sigma}) = 0$.
\item [(Q2)] $\nabla (f_{\vec{\sigma}}^j \oplus R(j, \vec{x}^j))(\vec{\sigma}) - \nabla f_{\vec{\sigma}}^j(\vec{\sigma}) = \vec{0}$.
    \item [(Q3)] The first coordinate of $\vec{x}$ is strictly positive. 
\end{itemize}

Suppose that we can indeed always find such an array $\vec{x}^j$. Then, players can construct a path from $\Gamma$ to a game in which $\vec{a}^{PI}$ is a Nash equilibrium by successively performing $R(j, \lambda \cdot \vec{x}^j)$ commitments for different values of $\lambda$ and $j$. By property (Q3), eventually the first coefficient of each component of $f_{\vec{\sigma}}$ becomes positive. When this happens, $\vec{a}^{PI}$ becomes a Nash equilibrium of the resulting game (note that we are assuming that $\vec{a}^{PI} = (a_1^1, a_2^1, \ldots, a_n^1)$). Moreover, property (Q2) guarantees that the determinant of $Df_{\vec{\sigma}}$ remains unaltered, and property (Q1) guarantees that $\vec{\sigma}$ remains a Nash equilibrium throughout the whole process. This would prove Theorem~\ref{sec:thm-bot} for $n \ge 3$.

Thus, it remains to show that an array $\vec{x}^j$ always exists. We will show this for $j = n+1$, which has the following expression:

$$
\begin{array}{lll}\
f_{\vec{\sigma}}^{n+1}(\vec{\sigma}) & = & (u_1(a_1^1, a_2^1, \ldots, a_n^1) - u_1(a_1^2, a_2^1, \ldots, a_n^1))\cdot p_2^1p_3^1\ldots p_n^1 \\
& + & (u_1(a_1^1, a_2^1, \ldots, a_n^2) - u_1(a_1^2, a_2^1, \ldots, a_n^2))\cdot p_2^1p_3^1\ldots p_n^2 \\
& + & (u_1(a_1^1, a_2^{N_2}, \ldots, a_n^{N_n}) - u_1(a_1^2, a_2^{N_2}, \ldots, a_n^{N_n}))\cdot p_2^{N_2}p_3^{N_3}\ldots p_n^{N_n}
\end{array}
$$

The proof for the other components is analogous. Our aim, thus, is to find an array $$\vec{x} = (x_{(1,1,\ldots, 1)}, x_{(1,1,\ldots, 2)}, \ldots, x_{(N_2, N_3, \ldots, N_n)})$$ such that $R(n+1, \vec{x})$ satisfies (Q1), (Q2) and (Q3). This is done in the following proposition.

\begin{proposition}\label{prop:magic-array}
Let $\vec{x}$ be an array defined as follows.

$$x_{(i_2,i_3, \ldots, i_n)} = 
\left\{
\begin{array}{ll}
0 & \mbox{if } i_j > 2 \mbox{ for some } j \in \{2,3,\ldots, n\}\\
(-1)^{(i_2 + \ldots + i_n + n + 1)} \cdot \sigma_2^{(3 - i_2)}\sigma_3^{(3 - i_3)} \ldots \sigma_n^{(3 - i_n)} & \mbox{otherwise.}
\end{array}
\right.$$

Then, $R(n+1, \vec{x})$ satisfies (Q1), (Q2) and (Q3). 
\end{proposition}

\begin{proof}
    \textbf{(Q1)} We have that $$
    (f_{\vec{\sigma}}^{n+1} \oplus R({n+1}, \vec{x}))(\vec{\sigma}) - f_{\vec{\sigma}}^{n+1}(\vec{\sigma}) = \sum_{(i_2, i_3, \ldots, i_n) \in \{1,2\}^{n-1}} (-1)^{(i_2 + \ldots + i_n + n + 1)} \cdot \sigma_2^{(3 - i_2)}\sigma_3^{(3 - i_3)} \ldots \sigma_n^{(3 - i_n)}  \sigma_2^{i_2}\sigma_3^{i_3} \ldots \sigma_n^{i_n},$$ which can be rearranged into $$(-1)^{n+1}\left(\prod_{j = 2}^n \sigma_j^1\sigma_j^2\right) \left(\sum_{(i_2, i_3, \ldots, i_n) \in \{1,2\}^{n-1}} (-1)^{(i_2 + \ldots + i_n)}\right).$$
    It is easy to check that the rightmost term is $0$.

    \textbf{(Q2)} This reduces to check that $$\frac{\partial\left((f_{\vec{\sigma}}^{n+1} \oplus R({n+1}, \vec{x})) - f_{\vec{\sigma}}^{n+1}\right)}{\partial p_i^j} (\vec{\sigma}) = 0$$ for all $i > 1$ and $j \in \{1,2\}$. We show this for $i = 2$ and $j = 1$. Since the function is symmetric on both $i$ and $j$, the general result will follow.

    We have that $$\frac{\partial\left((f_{\vec{\sigma}}^{n+1} \oplus R({n+1}, \vec{x})) - f_{\vec{\sigma}}^{n+1}\right)}{\partial p_2^1} (\vec{\sigma}) = \sum_{(i_3, \ldots, i_n) \in \{1,2\}^{n-2}} (-1)^{(i_3 + \ldots + i_n + n + 2)} \sigma_2^2\sigma_3^{(3 - i_3)}\ldots \sigma_n^{(3 - i_n)} \sigma_3^{i_3} \sigma_4^{i_4} \ldots \sigma_n^{i_n},$$
    which can be rearranged into $$(-1)^n \sigma_2^2 \left(\prod_{j = 3}^n \sigma_j^1\sigma_j^2\right) \left(\sum_{(i_3, \ldots, i_n) \in \{1,2\}^{n-2}} (-1)^{(i_3 + \ldots + i_n)}\right),$$
    which is also equal to $0$.

    \textbf{(Q3)} By construction, the sign of $x_{(1,1,\ldots,1)}$ is given by $(-1)^{2n}$, which is positive.
\end{proof}

This completes the proof of Theorem~\ref{thm:bot} in the case of three or more players.

\section{Proof of Theorem~\ref{thm:bot} for Two-Player Games with Binary Actions}\label{sec:2pl-2act}

The last remaining case is the one with two players and binary actions.
For simplicity, assume that $A_1 = A_2 = \{0,1\}$.

We begin the analysis with a characterization of all $2 
\times 2$ games in which there exists a Nash equilibrium with full support.

\begin{proposition}\label{prop:char-2by2}
Let $\Gamma = (P,A,U)$ be a normal-form game for two players in which $A_1 = A_2 = \{0,1\}$. Then, there exists a Nash equilibrium with full support if and only if, for each player $i$, one of the following conditions holds.

\begin{itemize}
    \item[(i)] Player $i$ is indifferent between $0$ and $1$ (i.e., $u_i(0_i, 0_{-i}) = u_i(1_i, 0_{-i})$ and $u_i(0_i, 1_{-i}) = u_i(1_i, 1_{-i})$).
    \item[(ii)] Player $i$ strictly prefers $0$ when $-i$ (i.e., the other player) prefers $0$ and $i$ strictly prefers $1$ when $-i$ prefers $1$.
    \item[(iii)] Player $i$ strictly prefers $0$ when $-i$ prefers $1$ and $i$ strictly prefers $1$ when $-i$ prefers $0$.
\end{itemize}
\end{proposition}

\begin{proof}
    As argued in Section~\ref{sec:degenerate}, $\Gamma$ has a Nash equilibrium with full support if and only if there exists a solution in $(0,1)^4$ of the following system of equations
    $$\begin{pmatrix}
        1 & 1 & 0 & 0\\
        d_1(0) & d_1(1)& 0 & 0 \\ 
        0 & 0 & 1 & 1 \\
        0 & 0 & d_2(0) & d_2(1)\\
    \end{pmatrix}
    \begin{pmatrix}
        p_2^0 \\ p_2^1 \\ p_1^0 \\ p_1^1
    \end{pmatrix}
    =
    \begin{pmatrix}
        1 \\ 0 \\ 1 \\ 0
    \end{pmatrix},
    $$
    where $d_i(a)$ denotes $u_i(0_i, a_{-i}) - u_i(1_i, a_{-i})$. This is equivalent to the following systems having solutions in $(0,1)^2$.
    $$\begin{pmatrix}
        1 & 1 \\
        d_1(0) & d_1(1)\\
    \end{pmatrix}
    \begin{pmatrix}
        p_2^0 \\ p_2^1
    \end{pmatrix}
    =
    \begin{pmatrix}
        1 \\ 0
    \end{pmatrix},
    $$

    $$\begin{pmatrix}
        1 & 1 \\
        d_2(0) & d_2(1)\\
    \end{pmatrix}
    \begin{pmatrix}
        p_1^0 \\ p_1^1
    \end{pmatrix}
    =
    \begin{pmatrix}
        1 \\ 0
    \end{pmatrix}.
    $$

    We will restrict our analysis to the first two coordinates of the solution, since the reasoning for the latter two is analogous. Since the probabilities are strictly positive, in order to have a solution in $(0,1)^2$, $d_1(0)$ and $d_1(1)$ must be either both zero or both non-zero and of different signs. In the first case, any solution in $(0,1)^2$ is valid. In the latter one we have that $p_2^0$ and $p_2^1$ are given by 

    $$\begin{array}{lll}
    p_2^0 & = & \frac{d_1(1)}{d_1(1) - d_1(0)}\\
    p_2^1 & = & \frac{-d_1(0)}{d_1(1) - d_1(0)}.\\
    \end{array}
    $$

    It is easy to check that these values are indeed in $(0,1)$ if $d_1(0)$ and $d_1(1)$ have different signs.
\end{proof}

This characterization shows that, if there exists a Nash equilibrium in $\Gamma$ with full support, then players must be of type (i) or (iii) (if any of them was of type (ii), $\vec{\sigma}$, would be degenerate). We next provide a commitment protocol $\pi^\delta_{2 \times 2}$ that allows players to converge on $\vec{a}^{PI}$. For simplicity, assume that $\vec{a}^{PI} = (0,0)$ and that $\delta < \min\{|u_i(1_i, 0_{-i}) - u_i(0_i, 0_{-i})|, |u_i(0_i,1_{-i}) - u_i(1_i, 1_{-i})|\}$ for all players $i$ of type (i) or (iii).  For completion, we will also say that a player is of type (iv) if it does not fall into any of the other three categories.

\textbf{Protocol $\pi^\delta_{2 \times 2}$:}

\begin{itemize}
    \item [\textbf{Step 1:}] Each player $i$ commits to pay $\delta$ in $(0_i, 1_{-i})$ and $(1_i, 1_{-i})$ in successive commitment steps, until their utility function satisfies $u_i(0_i,0_{-i}) > 
    max
    \{u_i(0_i, 1_{-i}), u_i(1_i, 1_{-i})\}$.
    \item [\textbf{Step 2:}] If player $i$ is of type (iii), in each commitment step, it commits to pay the following:

    $$
    \begin{array} {ll}
    \min\{\delta, u_i(1_i, 0_{-i}) - u_i(0_i, 0_{-i}) - \delta \} & \mbox{in } (1,0)\\
    \min\{\delta, u_i(0_i,1_{-i}) - u_i(1_i, 1_{-i}) - \delta \} & \mbox{in } (0,1).
        
    \end{array}
    $$

    Each player of type (iii) repeats this process until all players of this type satisfy that 
    $$u_i(1_i, 0_{-i}) - u_i(0_i, 0_{-i}) = u_i(0_i,1_{-i}) - u_i(1_i, 1_{-i}) = \delta.$$

    \item[\textbf{Step 3:}] Each player $i$ of type (iii) commits to pay $\delta$ on $(1,0)$ and $(0,1)$. Then, they vote to terminate the commitment phase.
\end{itemize}

Since players can only be of type (i) or (ii) on termination, it is easy to check that $(0,0)$ is a Nash equilibrium by the end of the protocol. It only remains to define what the players should do in case someone defects. The next proposition shows that, immediately after one of the players defects, they can always find a Nash equilibrium in which both of them are worse than in $(0,0)$. This shows that $\pi_{2 \times 2}$ can be extended to a subgame perfect equilibrium in $\Gamma^{\delta}$ in which players end up playing $(0,0)$.

\begin{proposition}\label{prop:2times2}
Let $\Gamma = (P,A,U)$ be a normal-form for two players with $A_1 = A_2 = \{0,1\}$ in which there exists a Nash equilibrium $\vec{\sigma}$ with full support and in which $(0,0)$ Pareto improves $\vec{\sigma}$. Suppose that players run protocol $\pi^\delta_{2 \times 2}$ and a player defects from the proposed strategy at some commitment step. Then, there exists a Nash equilibrium $\vec{\sigma}'$ in the resulting game $\Gamma' = (P,A,U')$ such that $u'_i(\vec{\sigma}') \le u_i(0,0)$ for all players $i$.
\end{proposition}

\begin{proof}
Note that the utility function of both players must satisfy that $u_i(0_i, 1_{-i}) \le u_i(0_i, 0_{-i})$, since otherwise it would contradict the fact that $(0,0)$ Pareto improves $\vec{\sigma}$. This means that, if after some commitment step, there exists a Nash equilibrium with full support, this Nash equilibrium is going to give less utility to the players than playing $(0,0)$ in the original game.

By construction, the type of the players is constant throughout the whole protocol with the exception of players of type (iii), who change to type (ii) at Step 3. Moreover, since $\delta < \min\{|u_i(1_i, 0_{-i}) - u_i(0_i, 0_{-i})|, |u_i(0_i,1_{-i}) - u_i(1_i, 1_{-i})|\}$ for all players $i$ of type (i) or (iii),  even if a player defects, she cannot change her type by doing so. Again, the only exception is at Step 3, where a player of type (iii) can choose to remain of type (iii) or convert to type (iv) instead of converting to type (ii). 
By Proposition~\ref{prop:char-2by2}, this means that, if a player defects before Step 3 or if a player defects in Step 3 without converting to type (iv), there always exists a Nash equilibrium $\vec{\sigma}'$ with full support in the resulting game, which by our previous argument satisfies $u'_i(\vec{\sigma}') \le u_i(0,0)$. The only remaining case is when a player $i$ defects in Step 3 and converts to Step (iv). In this case, because of Step 1 of $\pi^\delta_{2 \times 2}$, one of the following must hold:

\begin{itemize}
    \item $(0,0)$ is the outcome that gives the most utility to $i$.
    \item $u_i(0_i, 0_{-i}) < u_i(1_i, 0_{-i})$ and $u_i(0_i, 1_{-i}) = u_i(1_i, 1_{-i})$.
\end{itemize}

In the first case, any Nash equilibrium of the resulting game will suffice. In the latter, it can be easily checked that $(1,1)$ is a Nash equilibrium of the resulting game that satisfies $u'_i(1,1) < u_i(0,0)$. 

\end{proof}

\section{Proof of Theorem~\ref{thm:main}}\label{sec:thm-main}

The proof of Theorem~\ref{thm:main} is based on the one of Theorem~\ref{thm:bot}. In fact, the commitment protocol that we use for the proof consists of two stages. First, players select a maximizing welfare action $\vec{a}^{SW}$ and commit to transfer utilities in such a way that, by the end of this stage, each player $i$ gets exactly $x_i$ utility in $\vec{a}^{SW}$ and $\vec{\sigma}$ remains a Nash equilibrium of the resulting game that is Pareto improved by $\vec{a}^{SW}$. Second, players use the techniques of Section~\ref{sec:thm-bot} to make $\vec{a}^{SW}$ a Nash equilibrium of the resulting game. 

If we follow this approach, the proof of Theorem~\ref{thm:main} reduces to construct a commitment protocol that implements the first stage. Moreover, analogously to Proposition~\ref{prop:curve}, constructing such a commitment protocol further reduces to find a path as follows.

\begin{proposition}\label{prop:main-curve}
Suppose that there exists a path $\gamma$ from $\Gamma$ to another game $\Gamma' = (P,A,U')$ that satisfies the following conditions.

\begin{itemize}
    \item[(Q1)] $\gamma$ is continuous and finite.
    \item [(Q2)] For all action profiles $\vec{a}$, $w(\vec{a})$ (i.e., the social welfare on $\vec{a}$) weakly decreases through $\gamma$.
    \item [(Q3)] There exists an action profile  $\vec{a}^{SW}$ such that $u'_i(\vec{a}^{SW}) = x_i$ for all $i \in P$.
    \item [(Q4)] For all $\Gamma'' = (P,A,U'') \in \gamma$,
    $\vec{\sigma}$ is an equilibrium of $\Gamma''$.
    \item [(Q5)] For each player $i \in P$, the expected utility of playing $\vec{\sigma}$ weakly decreases through $\gamma$.
\end{itemize}

Then, there exists $D' > 0$ such that, for all $\delta \in (0, D')$, there exists a commitment protocol that goes from $\Gamma$ to $\Gamma'$ that satisfies the following conditions.

\begin{itemize}
    \item [(i)] At the beginning and after each commitment step, $\vec{\sigma}$ is a $(L_{\vec{\sigma}}^{\vec{x}}, \delta)$-strongly punishable Nash equilibrium of the resulting game, where $L_{\vec{\sigma}}^{\vec{x}} := \max_{i \in P}\{x_i - u_i(\vec{\sigma})\}$.
    \item [(ii)] Each player pays at most $\delta$ on each outcome in each commitment step.
\end{itemize}
\end{proposition}

The proof of Proposition~\ref{prop:main-curve} is pretty much identical to that of Proposition~\ref{prop:curve}.
We next that Theorem~\ref{thm:main} reduces to finding such a path $\gamma$. Suppose that there exists a path from $\Gamma$ to some other game $\Gamma'$ that satisfies (Q1), (Q2), (Q3) and (Q4) and (Q5).
Since $\vec{a}^{SW}$ strictly Pareto-improves $\vec{\sigma}$ in $\Gamma'$, we can apply Theorem~\ref{thm:bot} to deduce that there exists $D > 0$ such that, for all $\delta \in (0,D)$, there exists a subgame perfect Nash equilibrium in $(\Gamma')^\delta_\bot$ which players start on $\Gamma'$ and each player $i$ ends up receiving exactly $x_i$ utility. It is easy to check that this is also a subgame perfect Nash equilibrium in $(\Gamma')^\delta)$. Similarly, by Proposition~\ref{prop:main-curve}, there exists $D' > 0$ such that, for all $\delta \in (0,D)$, there is a commitment protocol that satisfies (i) and (ii). This means that, for all $\delta \in (0, \min\{D, D'\})$, players can run the protocol obtained from Proposition~\ref{prop:main-curve} to go from $\Gamma$ to $\Gamma'$, and then the protocols used in the proof of Theorem~\ref{thm:bot} to converge on a payoff profile that gives $x_i$ utility to each player $i$. It only remains to show that there always exists a punishment equilibrium, even if a player defects from the main protocol. Property (ii) of Proposition~\ref{prop:main-curve} guarantees its existence between $\Gamma$ and $\Gamma'$. After $\Gamma'$, the punishment equilibrium is provided by the protocols used in Section~\ref{sec:thm-bot}.

The argument above shows that Theorem~\ref{thm:main} reduces to find a path $\gamma$ that satisfies the properties of Proposition~\ref{prop:main-curve}. The following proposition shows how to construct such a path.

\begin{proposition}\label{prop:constructing-path-main}
Suppose without loss of generality that $\vec{a}^{SW} := (a_1^1, a_2^2, \ldots, a_n^2)$ is a welfare maximizing outcome of $\Gamma$. Consider the path $\gamma(\lambda)$ where the utilities of the players are defined by

$$(U(\lambda))_i(\vec{a}) := \left\{
\begin{array}{ll}
u_i(\vec{a}) + \lambda \cdot (x_i - u_i(\vec{a}^{SW})) & \mbox{if } u_i(\vec{a}^{SW}) > x_i \mbox{ or } \vec{a} \mbox{ is of the form } (a_i^j, \vec{a}^{SW}_{-i})\\
u_i(\vec{a}) - \lambda \frac{\sigma_{i+1}^1}{\sigma_{i+1}^2} (x_i - u_i(\vec{a}^{SW})) &  \mbox {if } u_i(\vec{a}^{SW}) < x_i \mbox { and } \vec{a} \mbox{ is of the form } (a_i^j, a_{i+1}^2, \vec{a}^{SW}_{-\{i, i+1\}})\\
u_i(\vec{a}) & \mbox{Otherwise}
\end{array}
\right.$$
for $\lambda \in [0,1]$, where we identify player $n+1$ with player $1$. This path satisfies (Q1), (Q2), (Q3), (Q4) and (Q5).
\end{proposition}

\begin{proof}
It is easy to check that properties (Q1) and (Q3) are satisfied by construction. To see that it also satisfies (Q2), note that the only outcomes $\vec{a}$ that may increase their social welfare are those of the form $(a_i^j, a_{-i}^{SW})$ for some $i \in P$ and $j \le N_i$. In order to check that the social welfare on these outcomes do not increase, note that the social welfare on a given outcome is a linear function, which means that we just have to check that the social welfare on these outcomes when $\lambda = 1$ is less or equal than when $\lambda = 0$. If we denote by $S(\vec{a})$ the set of players $k$ such that $u_i(\vec{a}^{SW}) < x_i$ and $\vec{a}$ is of the form $(a_i^j, \vec{a}^{SW}_{-i})$, we have that
$$\begin{array}{lcl}
\omega(1)(\vec{a}) - \omega(0)(\vec{a}) &=& \sum_{\{k \ : \ u_k(a^{SW}) > x_k\}} (x_k - u_k(\vec{a}^{SW})) + \sum_{k \in S(\vec{a})} (x_k - u_k(\vec{a}^{SW}))\\
& \le & \sum_{k \in P} (x_k - u_k(\vec{a}^{SW}))\\
& = & \left(\sum_{k \in P} x_k\right) - \left(\sum_{k \in P}u_k(\vec{a}^{SW})\right),
\end{array}$$
and both sums are equal since $\vec{a}^{SW}$ is a welfare maximizing outcome.

It is straightforward to check that $f_{\vec{\sigma}}$ is invariant, which implies that $\gamma$ satisfies (Q4). To see this, note that, if $u_i(\vec{a}^{SW}) > x_i$, then $i$ loses utility on all outcomes at the same rate. Meanwhile, if $u_i(\vec{a}^{SW}) < x_i$, $i$'s utility is only altered in outcomes of the form $(a_i^j, a_{i+1}^1, \vec{a}^{SW}_{\{i, i+1\}})$ and $(a_i^j, a_{i+1}^2, \vec{a}^{SW}_{\{i, i+1\}})$. It only remains to show that $\gamma$ satisfies (Q5). If $i$ is a player such that $u_i(\vec{a}^{SW}) > 
 x_i$, (Q5) is satisfied automatically since her utility in all outcomes decreases with $\lambda$. If $i$ is a player such that $u_i(\vec{a}^{SW}) < x_i$, we have that $$(U(\lambda)_i)(\vec{\sigma}) = u_i(\vec{\sigma}) + \lambda \cdot (x_i - u_i(\vec{a}^{SW})) \cdot \prod_{j \not = i} \sigma_j^1 - \lambda \frac{\sigma_{i+1}^1}{\sigma_{i+1}^2} (x_i - u_i(\vec{a}^{SW})) \cdot \sigma_{i+1}^2 \cdot \prod_{j \not \in \{i, i+1\}} \sigma_j^1 = u_i(\vec{\sigma}),$$ which is invariant.
\end{proof}

\section{Proof of Theorem~\ref{thm:strong-robust}}\label{sec:thm-strong}

Before proving Theorem~\ref{thm:strong-robust} we need the following notation. Given a game $\Gamma$, we denote by $\vec{x}_\Gamma$ be the ordered vector of all utilities in the game. For instance, if $\Gamma$ is a two-player game in which $A_1 = A_2 = \{0,1\}$, then $$\vec{x}_\Gamma = (u_1(0,0), u_1(0,1), u_1(1,0), u_1(1,1), u_2(0,0), u_2(0,1), u_2(1,0), u_2(1,1)).$$

Moreover, given a point $\vec{x} \in \mathbb{R}^N$, we denote by $\|x\|_2$, $\|x\|_1$ and $\|x\|_\infty$ its Euclidean, Manhattan and infinite norm, respectively. These are defined as follows:

$$\begin{array}{lll}
\|x\|_2 & = & (x_1^2 + x_2^2 + \ldots + x_N^2)^{\frac{1}{2}}\\
\|x\|_1 & = & |x_1| + |x_2| + \ldots + |x_N|\\
\|x\|_\infty & = & \max \{|x_1|, |x_2|, \ldots, |x_N|\}\\
\end{array}
$$

An important property of these norms is that they are \emph{equivalent}. This means that, for each pair of norms $\|\cdot\|_A$ and $\|\cdot\|_B$ from the ones above, there exists a constant $C_{A,B}$ such that $\|\vec{x}\|_A \le C_{A,B} \cdot  \|\vec{x}\|_B$ for all $\vec{x} \in \mathbb{R}^N$.

With this we proceed with the proof. We must also note that, to ease the notation, we will often use $\vec{\sigma}$ to denote the vector of probabilities by which players play their respective actions from their action sets when playing $\vec{\sigma}$ (as in the definition of $f_{\vec{\sigma}}$ and $r_{\vec{\sigma}}$ in Section~\ref{sec:degenerate}).

\begin{proof}
    Consider the function $F_{\vec{\sigma}} (\vec{x}, \vec{p}) : \mathbb{R}^{\left(n  \cdot \prod_{i = 1}^n |A_i| + \sum_{i = 1}^n |\sigma_i(A_i)|\right)} \to \mathbb{R}^{\sum_{i = 1}^n |\sigma_i(A_i)|}$, which is defined identically to $f_{\vec{\sigma}}$ except that we view the utilities of the players as additional variables, as opposed to fixed parameters (note that there are $n  \cdot \prod_{i = 1}^n |A_i|$ of such variables). By construction, we have that $F_{\vec{\sigma}} (\vec{x}_\Gamma, \vec{p}) = f_{\vec{\sigma}}(\vec{p})$, and in particular $F_{\vec{\sigma}}(\vec{x}_\Gamma, \vec{\sigma}) = 0$. Since $|Df_{\vec{\sigma}}| \not = 0$, by the implicit function theorem, there exists a $\mathcal{C}^\infty$ function $g : \mathbb{R}^{\left(n  \cdot \prod_{i = 1}^n |A_i|\right)} \to \mathbb{R}^{\sum_{i = 1}^n |\sigma_i(A_i)|}$ such that $g(\vec{x}_\Gamma) = \vec{\sigma}$ and  $F_{\vec{\sigma}}(\vec{x}, g(\vec{x})) = 0$ for all $\vec{x}$ in some neighborhood $V$ around $(\vec{x}_\Gamma, \vec{\sigma})$. 

    To show that $\vec{\sigma}$ is strongly robust, fix $\epsilon > 0$. Since the players' utilities is a continuous function of $\vec{\sigma}$, we have that there exists $\epsilon' > 0$ such that $\|\vec{\sigma}' - \vec{\sigma}\|_2 < \epsilon' \Longrightarrow |u_i(\vec{\sigma}) - u_i(\vec{\sigma}')| < \epsilon/2$. Let $\delta_1 > 0$ be a value such that $\|\vec{x}_\Gamma - \vec{x}\|_2 < \delta_1 \Longrightarrow \|g(\vec{x}_\Gamma) - g(\vec{x}')\|_2 < \epsilon'$ (which exists since $g$ is continuous), and let $\delta_2 > 0$ be another value such that
    $\|\vec{x}\|_2 < \delta_2 \Longrightarrow \|\vec{x}\|_1 < \epsilon/2$. 
    Let $\delta' = \min(\delta_1, \delta_2)$ and let $\delta > 0$ be a value such that $\|\vec{x}\|_\infty < \delta \Longrightarrow \|\vec{x}\|_2 < \delta'$. Note that $\delta_2$ and $\delta$ exist since the infinite, Manhattan and Euclidean norms are equivalent.
    
    Let $\Gamma'$ be a normal-form game such that $d(\Gamma, \Gamma') < \delta$ and let $u'_i$ denote the utility of player $i$ in $\Gamma'$. Since $\Gamma'$ satisfies that $d(\Gamma, \Gamma') < \delta$, we have that $\|\vec{x}_\Gamma - \vec{x}_{\Gamma'}\|_\infty < \delta$, and therefore that $\|\vec{x}_\Gamma - \vec{x}_{\Gamma'}\|_2 < \delta'$. This in turn implies that $\|g(\vec{x}_\Gamma) - g(\vec{x}_{\Gamma'})\|_2 < \epsilon'$, which is equivalent to $\|\vec{\sigma} - g(\vec{x}_{\Gamma'})\|_2 < \epsilon'$, and therefore $|u_i(\vec{\sigma}) - u_i(g(\vec{u}_{\Gamma'}))| < \epsilon/2$. It follows that $$|u_i(\vec{\sigma}) - u_i'(g(\vec{u}_{\Gamma'}))| \le |u_i(\vec{\sigma}) - u_i(g(\vec{u}_{\Gamma'}))| + |u_i(g(\vec{u}_{\Gamma'})) - u'_i(g(\vec{u}_{\Gamma'}))| < \frac{\epsilon}{2} + |(u_i - u'_i)(g(\vec{u}_{\Gamma'}))|.$$

    To complete the proof, we have to show that $|(u_i - u'_i)(g(\vec{u}_{\Gamma'}))| \le \epsilon/2$. Note that $$|(u_i - u'_i)(g(\vec{u}_{\Gamma'}))| \le \|\vec{x}_\Gamma - \vec{x}_\Gamma'\|_1$$ since all the components of $\vec{\sigma}$ are values between $0$ and $1$. Since $\|\vec{x}_\Gamma - \vec{x}_\Gamma'\|_2 < \delta \le \delta_2$, we have that $\|\vec{x}_\Gamma - \vec{x}_\Gamma'\|_1 < \epsilon/2$, which implies that $$|(u_i - u'_i)(g(\vec{u}_{\Gamma'}))| \le \epsilon/2,$$
    as desired.
    \end{proof}

\end{document}